\documentclass[sn-mathphys-num]{sn-jnl}

\usepackage[T1]{fontenc}
\usepackage{graphicx}
\usepackage{amsmath,amssymb,amsfonts,amsthm,mathtools,mathrsfs}
\usepackage[title]{appendix}
\usepackage{booktabs,enumitem,xcolor}

\theoremstyle{thmstyleone}%
\newtheorem{theorem}{Theorem}%
\newtheorem{proposition}[theorem]{Proposition}%
\newtheorem{lemma}[theorem]{Lemma}%
\newtheorem{corollary}[theorem]{Corollary}%
\theoremstyle{thmstyletwo}%
\newtheorem{remark}{Remark}%
\theoremstyle{thmstylethree}%
\theoremstyle{thmstyleone}%
\newcommand{\restatename}{Theorem}%
\newtheorem*{restatement}{\restatename}%
\begin{document}
	
	\title[Fractality of the Schr\"odinger Density]{Fractality of the
		Schr\"odinger Density for\\ Rough Data}
	
	\author*[1]{\fnm{Masoud} \sur{Ataei}}\email{masoud.ataei@utoronto.ca}
	
	\affil[1]{\normalsize\orgdiv{Department of Mathematical and
			Computational Sciences},\\ \orgname{University of Toronto},
		\country{Canada}}
	
	\abstract{The observable of the Talbot effect is the intensity
		behind the grating: the density of the evolving field, not the
		field itself. Its fractality was previously known only for step
		data with rational jumps. We prove the general case. For arbitrary
		real data of bounded variation with at least one jump, the density
		of the free Schr\"odinger evolution on the torus is fractal at
		almost every time, with graph of upper box dimension exactly three
		halves. Its critical Sobolev mass diverges logarithmically, at a
		rate given in closed form by Wiener's jump statistic of the datum
		and independent of the positions and phases of the jumps. The same
		rate is halved along time traces and doubled for the Airy flow, so
		the critical mass probes the dispersion relation. Transported by
		the rank calculus of derangetropy operators, the theory equips
		every probability law with a quantum carpet: Gauss-sum revivals on
		quantile cells at rational times, a universal fractal at almost
		every other. The spectral growth rate of a single measured
		intensity profile returns the jump content of the grating, a
		prediction open to test in optical and matter-wave
		interferometry.}
	
	\keywords{Talbot Effect, Dispersive Quantization, Fractal Solutions,
		Schr\"odinger Equation, Bounded Variation, Wiener's Theorem, Quantum
		Carpets, Rank Transformations, Derangetropy}
	
	\maketitle
	
	\newpage
	
	\section{Introduction}\label{sec:intro}
	
	When a plane wave passes a periodic grating, the intensity pattern
	behind the grating repeats itself at a fixed distance and, between
	repetitions, weaves a self-similar tapestry: the Talbot effect of
	wave optics, whose mathematical model is the free Schr\"odinger
	evolution of rough periodic data. The dichotomy is now classical on
	the side of the \emph{field}. At rational multiples of the period
	the evolution of a step function is again a step function: a
	quantization, or revival, phenomenon governed by quadratic Gauss
	sums~\cite{talbot1836,berry1996,berryklein1996,oskolkov1992}. At
	irrational times the field is a continuous, nowhere differentiable
	function. For data of bounded variation its graph has box
	dimension exactly three halves at almost every time, a theorem of
	Rodnianski~\cite{rodnianski2000} building
	on~\cite{kapitanski1999,oskolkov1992}, with refinements and
	extensions by Erdo\u{g}an and
	Tzirakis~\cite{erdogan2016,erdogan2019}, to higher dimensions by
	Erdo\u{g}an, Huynh, and McConnell~\cite{erdogan2024}, and to
	Schr\"odinger evolutions with potentials by Cho, Kim, Kim, Kwon,
	and Seo~\cite{cho2024}.
	
	The observable, however, is not the field. What an optical or
	matter-wave experiment records is the intensity, the squared
	modulus of the field. For the intensity the fractality question
	has remained open except in one special case: for step data with
	jumps at rational points, Chousionis, Erdo\u{g}an, and
	Tzirakis~\cite{chousionis2015} proved density fractality through
	the revival arithmetic of Gauss sums, and stated the general case
	as beyond that method. The obstruction is structural. Unitarity
	controls the field and not its modulus. The conserved Sobolev
	roughness of the field could in principle hide in the phase and
	cancel in the intensity, and no smoothing estimate rules this
	out.
	
	This paper resolves the general case. For real data of bounded
	variation with at least one jump, we prove a strong law of large
	numbers for the critical Sobolev mass of the evolved density. The
	partial sums that measure the density's mass at the critical
	Sobolev regularity of exponent one half grow logarithmically at
	almost every time. The almost-sure rate is given in closed form by
	Wiener's quantitative theorem on the jump content of a function of
	bounded variation~\cite{wiener1924,zygmund2002}: the squared norm
	of the datum times the summed squares of its jumps, divided by
	twice the square of pi. At almost every time the density therefore
	lies outside the critical Sobolev space. Combined with a smoothing
	estimate placing it just below H\"older regularity one half, its
	graph has upper box dimension exactly three halves.
	
	The mechanism is arithmetic rather than analytic. Within each
	Fourier mode of the density the time frequencies are
	collision-free, so the time average of the mode's squared modulus
	is a phase-blind positive sum that no arrangement of the jumps can
	reduce. Wiener's theorem converts the weighted sum of these
	averages into the logarithmic law. Across modes, collisions are
	confined to a sparse lattice, so the variance stays bounded while
	the mean diverges, and the Borel--Cantelli lemma upgrades the mean
	divergence to an almost-sure law. The phases keep their freedom at
	every fixed scale and lose it in the aggregate. The same collision
	arithmetic computes the dispersion dependence of the constant.
	Along time traces the quadratic frequency scaling halves it, and
	for the Airy flow the paired collision classes double it: the
	critical mass of the density is an exactly measurable probe of
	resonance structure.
	
	The second circle of results transports this theory onto probability
	laws. A derangetropy operator, the term a portmanteau of
	derangement and entropy coined in~\cite{ataei2024} for functionals
	tracking the cyclical rearrangement of information in a law,
	reweights a probability density by a fixed profile of its own
	cumulative distribution function. Three operators of this form were
	introduced in~\cite{ataei2024,ataei2025} as models of cyclical
	information flow, and the family is characterized and developed as
	a calculus in~\cite{ataei2026}. The preliminaries the present
	paper needs, the operators, their transport formula, and the
	canonical kernel, are recalled in Section~\ref{sec:prelim}. Evolving the canonical modulation profile,
	the squared ground state of the Dirichlet Laplacian of rank space,
	under the Schr\"odinger dynamics of the unit interval yields the
	\emph{carpet flow}: a periodic curve of modulations that prints the
	Talbot dynamics on every absolutely continuous law at once. At
	rational times the modulated law revives into a quadratic Gauss-sum
	histogram on quantile cells. At almost every time it weaves a
	fractal of upper box dimension exactly three halves, in every
	quadrature of the amplitude and, by the strong law, for the density
	itself. And the information divergence from the base law breathes
	periodically, with universal time-averaged chi-square level and no
	long-run entropy production.
	
	Rank space is thus a quantum box that every probability law
	carries. The results assert that the intensity carpet behind an
	\emph{arbitrary} grating of bounded variation, irregular slit
	positions, rough transmission profiles, has upper box dimension
	exactly three halves at almost every distance, its critical mass
	diverging at the rate set by Wiener's jump statistic of the
	grating. The prediction is open to direct test on measured carpets
	in optical or matter-wave interferometry.
	
	The remainder of the paper is organized as follows.
	Section~\ref{sec:prelim} recalls the rank calculus and constructs
	the unitary lift and the torus dictionary.
	Section~\ref{sec:carpet} develops the carpet flow: the observable
	clock, the quantile mosaics of rational times, the information
	ledger, and the universal carpet theorem.
	Section~\ref{sec:stronglaw} states and proves the strong law for
	the critical mass, its universality over seeds and laws, and the
	collision arithmetic of the dispersion. We conclude in
	Section~\ref{sec:conclusion}; proofs are collected in
	Appendix~\ref{app:proofs}.
	
	\section{Preliminaries}\label{sec:prelim}
	
	Write $\mathcal{D}_{+}$ for the class of probability laws on
	$\mathbb{R}$ whose density $f$ is continuous and positive on the
	interior of an interval support; $F$ denotes the distribution
	function, $Q = F^{-1}$ the quantile function, and $z = F(x)$ the
	\emph{rank} of the point $x$. By the probability integral
	transform, $F(X)$ is uniform on $(0,1)$ for $X \sim f$: in its own
	rank coordinate every law of $\mathcal{D}_{+}$ is uniform, and for
	every integrable $\Psi$ on $(0,1)$,
	\begin{equation}
		\int_{\mathbb{R}}\Psi\big(F(x)\big)\,f(x)\,dx
		= \int_{0}^{1}\Psi(z)\,dz,
		\label{eq:universality}
	\end{equation}
	the \emph{universality principle}: integrals of rank functionals do
	not depend on the law. A \emph{kernel} is a probability density $w$
	on $[0,1]$, and the \emph{derangetropy operator} with kernel $w$
	acts on laws by
	\begin{equation}
		\rho_{w}[f] = w(F)\,f,
		\label{eq:operator}
	\end{equation}
	a probability density again by the identity~\eqref{eq:universality};
	the construction originates in the information-theoretic
	functionals of~\cite{ataei2024,ataei2025}.
	These operators compose through the interval maps
	$A_{w}(z) = \int_{0}^{z}w$, by the \emph{transport formula}
	\begin{equation}
		F_{\rho_{w}[f]} = A_{w}\circ F,
		\label{eq:transport}
	\end{equation}
	obtained by integrating the relation~\eqref{eq:operator}. The
	composition of two operators is again a derangetropy operator,
	with interval map the composition of the two maps, so the dynamics
	of the calculus is exactly solvable. Of the calculus developed
	in~\cite{ataei2026} the present paper uses only the transport
	formula~\eqref{eq:transport}, the universality
	principle~\eqref{eq:universality}, and the canonical kernel now
	described. The \emph{canonical kernel} is the squared ground
	state of the Dirichlet Laplacian of rank space: writing
	\begin{equation}
		H_{0} := -\tfrac12\,\frac{d^{2}}{dz^{2}}
		\qquad\text{on } (0,1),\ \text{Dirichlet conditions},
		\label{eq:H0}
	\end{equation}
	with eigenfunctions $\varphi_{k} = \sqrt2\sin(k\pi z)$ and
	eigenvalues $k^{2}\pi^{2}/2$, the canonical kernel is
	\begin{equation}
		w(z) = |\varphi_{1}(z)|^{2} = 2\sin^{2}(\pi z),
		\label{eq:canonical-kernel}
	\end{equation}
	and the canonical amplitude $\sin(\pi z) = \sqrt{w/2}$ is, up
	to normalization, the ground state itself.
	
	The object of this paper is the \emph{unitary lift} of the
	modulation. Let $\mathbf{1}$ denote the amplitude of the uniform
	kernel, the constant function $1$ on $(0,1)$, and let
	\begin{equation}
		u_{\tau} := e^{-i\tau H_{0}}\mathbf{1}
		\label{eq:unitary-lift}
	\end{equation}
	be its Dirichlet evolution. At every $\tau$ the function
	$w_{\tau} := |u_{\tau}|^{2}$ is a kernel, since
	$\int_{0}^{1}w_{\tau} = \Vert u_{\tau}\Vert_{L^{2}}^{2}
	= \Vert\mathbf{1}\Vert_{L^{2}}^{2} = 1$ by unitarity: mass
	conservation is built into the geometry rather than imposed. The
	odd extension of $\mathbf{1}$ to the doubled circle
	$\mathbb{T}_{2} := \mathbb{R}/2\mathbb{Z}$ is the square wave,
	the initial datum of the Talbot effect. The eigenfunction
	expansion of the evolution then becomes the free Schr\"odinger
	evolution of the square wave on $\mathbb{T}_{2}$; the precise
	dictionary is the display~\eqref{eq:carpet-seed} below. More generally, for $g$ of
	bounded variation on the standard torus with Fourier coefficients
	$\hat g(n)$ and jumps $\{J_{j}\}$, we write throughout
	\begin{equation}
		u(x,t) = \sum_{n}\hat g(n)\,e^{2\pi i(nx - n^{2}t)},
		\qquad
		w = |u|^{2}
		\label{eq:torus-evolution}
	\end{equation}
	for the free evolution and its density, and $P_{N}$ for the dyadic
	Fourier block retaining the frequencies $|n| \in [N, 2N)$.

	\section{Quantum Carpets}\label{sec:carpet}
	
	This section develops the unitary axis of the rank calculus, on
	which the modulation kernel itself evolves under the
	Schr\"odinger dynamics of rank space. With the unitary
	lift~\eqref{eq:unitary-lift} of Section~\ref{sec:prelim}, define
	the \emph{carpet flow} of
	a law $f \in \mathcal{D}_{+}$ as
	\begin{equation}
		f_{\tau} := \rho_{w_{\tau}}[f]
		= \big|u_{\tau}(F)\big|^{2}\,f,
		\qquad
		w_{\tau} = |u_{\tau}|^{2}:
		\label{eq:carpet-flow}
	\end{equation}
	the carpet flow does nothing but let the Dirichlet eigenfunction
	expansion of the uniform amplitude rotate at its natural
	frequencies, printed on $f$ through its ranks. At $\tau = 0$ the
	expansion is coherent and $f_{0} = f$. At every $\tau$ the kernel
	$w_{\tau}$ has unit mass by unitarity, as noted in
	Section~\ref{sec:prelim}, so each $f_{\tau}$ is a law
	automatically. The carpet flow is thus a closed curve inside the modulation family, traversed at constant unitary
	speed. On the doubled circle
	$\mathbb{T}_{2} := \mathbb{R}/2\mathbb{Z}$, where the odd extension
	of $\mathbf{1}$ is the square wave, the seed and its evolution read
	\begin{equation}
		\begin{gathered}
			u_{s}(z) = \sum_{n\ \mathrm{odd}}c_{n}\,
			e^{-2\pi in^{2}s}\,e^{i\pi nz},\\[2pt]
			c_{n} := \frac{2}{i\pi n},
			\qquad
			s := \frac{\pi\tau}{4},
		\end{gathered}
		\label{eq:carpet-seed}
	\end{equation}
	the rescaled time $s$ making the mode-$n$ phase
	$e^{-2\pi in^{2}s}$. This is the initial datum of the Talbot
	effect of wave optics, carried by rank space, and the dilation
	prints its evolution on every law at once. Throughout this section $P_{N}$ denotes the
	dyadic block of the expansion~\eqref{eq:carpet-seed}, retaining the
	frequencies $|n| \in [N, 2N)$.
	
	\begin{proposition}[The observable clock]\label{prop:carpet-clock}
		The amplitude flow $s \mapsto u_{s}$ has period $1$, so the flow
		of laws is periodic in $\tau$ with period $4/\pi$; its minimal
		period is $1/(2\pi)$, the density flow $s \mapsto w_{s}$ having
		minimal period $\tfrac18$: for every law and every $\tau$,
		\begin{equation}
			\begin{gathered}
				f_{\tau + \frac{1}{2\pi}} = f_{\tau},\\[2pt]
				u_{s+1/8} = e^{-i\pi/4}\,u_{s}.
			\end{gathered}
			\label{eq:carpet-period}
		\end{equation}
		Within one period the density flow is a palindrome,
		$w_{1/8 - s} = w_{s}$, and at every time the kernel is symmetric
		about the median rank, $w_{s}(1-z) = w_{s}(z)$.
	\end{proposition}
	
	\begin{proof}
		The amplitude period is $e^{-2\pi in^{2}} = 1$. The seed populates
		only odd modes, and odd squares are $1$ modulo $8$, so advancing $s$
		by $\tfrac18$ multiplies every mode by the common eighth root of
		unity $e^{-i\pi/4}$, which is invisible in the modulus. The
		observable flow thus lives on a circle eight times shorter than
		its unitary dressing, the discrepancy carried by exactly the roots
		of unity that govern quadratic Gauss sums. For minimality, the
		coefficient of $e^{2\pi i\cdot 2z}$ in $w_{s}$ is
		\begin{equation}
			\widehat{w_{s}}(2) = \sum_{n\ \mathrm{odd}}
			\frac{4}{\pi^{2}n(n-2)}\;e^{-8\pi i(n-1)s},
			\label{eq:mode-two}
		\end{equation}
		a series in $s$ whose frequencies $4(n-1)$ are pairwise distinct
		and whose coefficients never vanish; a period $p$ of the density
		flow equates the almost periodic function~\eqref{eq:mode-two} with
		its translate, and since the frequencies and Fourier--Bohr
		coefficients of an almost periodic function are unique, this
		forces $4(n-1)p \in \mathbb{Z}$ for every odd $n$, and $n = 3$
		gives $8p \in \mathbb{Z}$. Since $\overline{c_{n}} = c_{-n}$ for the coefficients of the
		seed~\eqref{eq:carpet-seed}, conjugating the
		series gives $\overline{u_{s}} = u_{-s}$,
		whence $w_{-s} = w_{s}$ and, with the
		relation~\eqref{eq:carpet-period}, the palindrome. The median
		symmetry is the substitution $z \mapsto 1-z$, which multiplies the
		mode $n$ by $e^{i\pi n} = -1$ and is undone by $n \mapsto -n$.
	\end{proof}
	
	Figure~\ref{fig:carpet} displays one observable period of the
	flow, on rank space and transported onto the Gaussian law: the
	palindrome of the clock, the bright revivals of rational times, and
	the fractal interstices between them are visible to the eye.
	
	\begin{figure}[t]
		\centering
		\includegraphics[width=\textwidth]{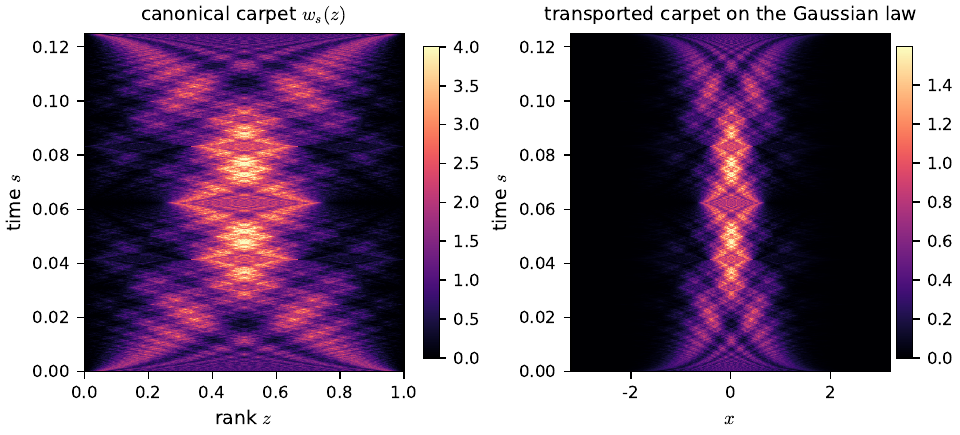}
		\caption{One observable period of the carpet flow, computed from
			the series~\eqref{eq:carpet-seed}. Left: the canonical carpet, the
			kernel $w_{s}$ over rank $z$ and rescaled time $s \in [0,\tfrac18]$;
			rational times print quantile mosaics, and the flow is palindromic
			about $s = \tfrac{1}{16}$. Right: the same period transported onto
			the standard Gaussian law, the density
			$f_{\tau}(x) = w_{s}(F(x))\,f(x)$: every law carries the same
			carpet in its own quantile coordinates.}\label{fig:carpet}
	\end{figure}
	
	At rational times the carpet flow revives into finitely many
	quantile cells.
	
	\begin{theorem}[Quantile mosaics]\label{thm:quantile-mosaic}
		Let $s = p/q$ in lowest terms. The phase
		$\chi(n) = e^{-2\pi in^{2}p/q}$ is periodic in $n$ with period $q$,
		and its discrete Fourier expansion yields the revival
		\begin{equation}
			\begin{gathered}
				u_{p/q}(z) = \sum_{r=0}^{q-1}\hat c(r)\,
				u_{0}\Big(z - \frac{2r}{q}\Big),\\[2pt]
				\hat c(r) = \frac{1}{q}\sum_{n=0}^{q-1}
				e^{-2\pi i\left(pn^{2} - rn\right)/q},
			\end{gathered}
			\label{eq:carpet-revival}
		\end{equation}
		a finite superposition of translated square waves whose
		coefficients are generalized quadratic Gauss
		sums~\cite{berndt1998}. Consequently $w_{p/q}$ is constant on the
		cells of the quantile lattice $\{l/q\}$, and for every law $f$,
		\begin{equation}
			f_{(p/q)\cdot 4/\pi} = w_{p/q}(F)\,f,
			\label{eq:mosaic-histogram}
		\end{equation}
		a histogram reweighting of $f$ on at most $q$ of its own quantile
		cells, of total mass one. In particular the rational-time orbit of
		every law consists of explicitly computable histogram modulations
		of that law.
	\end{theorem}
	
	\begin{proof}
		Periodicity: $(n+q)^{2} - n^{2} = 2nq + q^{2}$, and
		$p(2nq + q^{2})/q = 2np + pq \in \mathbb{Z}$ for every parity of
		$q$. Expanding $\chi$ in the characters of
		$\mathbb{Z}/q\mathbb{Z}$,
		\begin{equation*}
			\chi(n) = \sum_{r=0}^{q-1}\hat c(r)\,e^{-2\pi irn/q}
		\end{equation*}
		with the coefficients of the
		superposition~\eqref{eq:carpet-revival}, and substituting into the
		series~\eqref{eq:carpet-seed} converts the character
		$e^{-2\pi irn/q}$ into the translation by $2r/q$ on
		$\mathbb{T}_{2}$, since
		$e^{-2\pi irn/q}\,e^{i\pi nz} = e^{i\pi n(z - 2r/q)}$, which is
		the superposition formula~\eqref{eq:carpet-revival}. The jumps of $u_{0}$ lie on the
		integers, so those of the superposition lie on the lattice
		$\{2r/q\} \cup \{1 + 2r/q\}$ modulo $2$, which meets $(0,1)$ inside
		$\{l/q\}$; the squared modulus inherits the cells, and unit mass is
		the unitarity of the dilation.
	\end{proof}
	
	Two instances display the arithmetic. At $s = \tfrac13$ the weight
	takes the values
	\begin{equation}
		w_{1/3} = \Big(\tfrac13,\ \tfrac73,\ \tfrac13\Big)
		\quad\text{on the tercile cells:}
		\label{eq:tercile-mosaic}
	\end{equation}
	at time $\tau = 4/3\pi$ every law concentrates seven ninths of its
	mass on its middle tercile. At $s = \tfrac{1}{16}$, the phases
	depend on $n$ only through the Jacobi symbol,
	$\chi(n) = e^{-i\pi/8}\left(\tfrac{2}{n}\right)$, and summing the
	four translates of the Gauss expansion modulo $8$ gives exactly
	\begin{equation}
		\begin{gathered}
			u_{1/16}\big|_{(0,1)}
			= e^{-i\pi/8}\,\sqrt2\;\mathbf{1}_{(\frac14,\frac34)},\\[2pt]
			w_{1/16} = \big(0,\ 2,\ 2,\ 0\big)
			\quad\text{on the quartile cells}:
		\end{gathered}
		\label{eq:quartile-blackout}
	\end{equation}
	one sixteenth of the way through its period, every law
	extinguishes its outer quartiles and doubles its interquartile
	body; the mosaic is governed by the quadratic residues modulo
	$8$. The
	amplitude in the display~\eqref{eq:quartile-blackout} is, up to
	phase and recentering, the balanced two-valued step
	$g = \mathbf{1}_{(1/4,3/4)} - \mathbf{1}_{(0,1/4)\cup(3/4,1)}$,
	through which the carpet flow passes at the sixteenth of its
	period. Figure~\ref{fig:mosaics} shows this blackout, a richer
	mosaic, and a generic fractal section of the carpet.
	
	\begin{figure}[t]
		\centering
		\includegraphics[width=\textwidth]{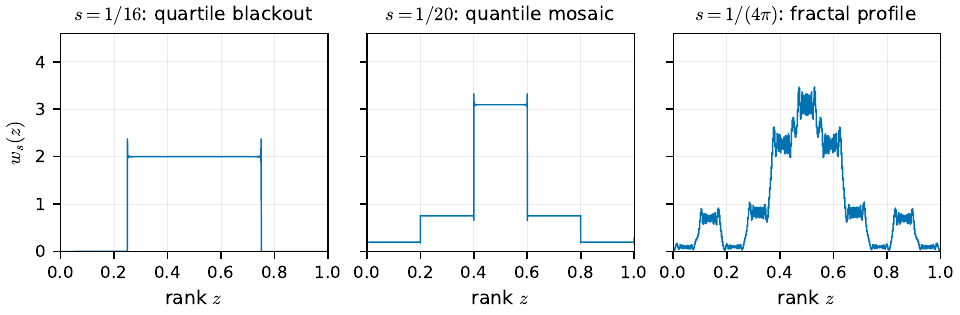}
		\caption{Three sections of the carpet, computed from the truncated
			series~\eqref{eq:carpet-seed}. Left: at $s = \tfrac{1}{16}$ the
			density is the quartile blackout, vanishing identically on the
			outer quartiles and equal to $2$ on the interquartile range.
			Center: the five-cell quantile mosaic at $s = \tfrac{1}{20}$; the
			mosaic lives on the fifths rather than the twentieths, the Gauss
			coefficients with odd $r$ vanishing for $q \equiv 0 \bmod 4$, so
			the cell count of Theorem~\ref{thm:quantile-mosaic} is not sharp.
			Right: at the generic time $s = 1/(4\pi)$ the section is a fractal
			profile of upper box dimension $\tfrac32$.}\label{fig:mosaics}
	\end{figure}
	
	Before the fractal theory, the information ledger of the unitary
	axis, which is exactly computable and exactly bounded.
	
	\begin{theorem}[Bounded breathing of information]\label{thm:carpet-information}
		For every $f \in \mathcal{D}_{+}$:
		\begin{enumerate}[label=\normalfont(\roman*), leftmargin=*]
			\item the time-averaged carpet is the tent modulation,
			\begin{equation}
				\begin{gathered}
					2\pi\int_{0}^{1/2\pi} f_{\tau}\,d\tau
					= \rho_{\Lambda}[f],\\[2pt]
					\Lambda(z) = 4\min(z,\,1-z);
				\end{gathered}
				\label{eq:carpet-tent}
			\end{equation}
			\item the time-averaged chi-square divergence from the base law is
			universal,
			\begin{equation}
				\begin{gathered}
					2\pi\int_{0}^{1/2\pi}
					\chi^{2}\big(f_{\tau}\,\|\,f\big)\,d\tau
					= \frac{2}{3},\\[2pt]
					\chi^{2}\big(f_{\tau}\|f\big)
					= \int_{0}^{1}\big(w_{s} - 1\big)^{2}dz;
				\end{gathered}
				\label{eq:carpet-chi}
			\end{equation}
			\item information breathes but never grows:
			$D(f_{\tau}\,\|\,f) \le \log\big(1 + \chi^{2}(f_{\tau}\|f)\big)$,
			and $\sup_{\tau}\chi^{2}(f_{\tau}\|f)$ is a finite absolute
			constant.
		\end{enumerate}
	\end{theorem}
	
	\begin{proof}
		All three statements live in rank space, since by the universality
		principle~\eqref{eq:universality} the divergences depend only on
		$w_{s}$. For part (i), the time average retains exactly the static
		frequencies: the coefficient of $e^{2\pi imz}$ in $w_{s}$ is
		$\sum_{n} c_{n}\overline{c_{n-2m}}\,e^{-2\pi i\cdot 4m(n-m)s}$, and
		the frequency $4m(n-m)$ vanishes precisely at $n = m$, admissible
		only for odd $m$, with static weight
		$c_{m}\overline{c_{-m}} = -4/(\pi^{2}m^{2})$. Hence
		\begin{equation*}
			\int_{0}^{1}w_{s}\,ds\;(z)
			= 1 - \frac{8}{\pi^{2}}\sum_{m\ \mathrm{odd},\,\ge1}
			\frac{\cos(2\pi mz)}{m^{2}}
			= 4\min(z,\,1-z),
		\end{equation*}
		the classical triangular-wave series. For part (ii), expand
		$\int_{0}^{1}w_{s}^{2}\,dz$ by Parseval and average in $s$: a
		quadruple $(n_{1},n_{2},n_{3},n_{4})$ of odd modes survives both
		the spatial pairing and the time average precisely when
		$n_{1} + n_{3} = n_{2} + n_{4}$ and
		$n_{1}^{2} + n_{3}^{2} = n_{2}^{2} + n_{4}^{2}$, and equal sums
		with equal sums of squares force equal pairs
		$\{n_{1},n_{3}\} = \{n_{2},n_{4}\}$. The two pairings give
		\begin{equation*}
			\int_{0}^{1}\!\!\int_{0}^{1}w_{s}^{2}\,dz\,ds
			= 2\Big(\sum_{n}|c_{n}|^{2}\Big)^{2} - \sum_{n}|c_{n}|^{4}
			= 2 - \frac13 = \frac53,
		\end{equation*}
		since $\sum|c_{n}|^{2} = 1$ and
		$\sum|c_{n}|^{4} = (32/\pi^{4})\sum_{\mathrm{odd}}n^{-4} = \tfrac13$;
		subtracting the squared mass gives the value~\eqref{eq:carpet-chi}.
		For part (iii), Jensen's inequality with respect to the probability
		measure $w_{s}\,dz$ gives
		$\int w_{s}\log w_{s} \le \log\int w_{s}^{2}$, and the uniform
		bound follows from
		$\int_{0}^{1}w_{s}^{2}\,dz
		\le \sum_{k}\big(\sum_{n}|c_{n}c_{n-k}|\big)^{2}$, a convergent
		series independent of $s$.
	\end{proof}
	
	\begin{remark}[No entropy production]
		On the unitary axis, information breathes but never grows. The
		divergence from the base law is a bounded palindromic function of
		time, with the universal mean-square level $\tfrac23$ of the
		identity~\eqref{eq:carpet-chi}, so complexity cannot appear as
		entropy production; the theorems below show that it appears
		instead as geometry, as fractal dimension. The tent
		identity~\eqref{eq:carpet-tent} is a point of contact with the
		canonical kernel: the tent $\Lambda$ has the same unit mass, the
		same endpoint zeros, and the same central peak of height two, so
		the time-averaged carpet is a canonical-type modulation sharpened
		to a corner.
	\end{remark}
	
	At irrational times the mosaics give way to fractals, and the
	remainder of the section proves the dichotomy in its sharp form.
	The two mechanisms are conservation and generic smoothing. On one
	hand the evolution is a unimodular Fourier multiplier, so
	$|\widehat{u_{s}}(n)| = 2/(\pi|n|)$ for every $s$ and every Sobolev
	norm is frozen. In particular
	$\sum_{|n|\le N}|n|\,|\widehat{u_{s}}(n)|^{2} \asymp \log N$,
	so $u_{s}$ lies outside $H^{1/2}(\mathbb{T}_{2})$ at every time:
	the roughness of the seed's jump is indestructible. On the other
	hand, at almost every time the roughness is spread rather than
	concentrated.
	
	\begin{lemma}[Generic smoothing]\label{lem:carpet-smoothing}
		For almost every $s$: $u_{s}$ is continuous, and for every
		$\varepsilon > 0$ there is $C_{s,\varepsilon}$ with
		\begin{equation}
			\Vert P_{N}u_{s}\Vert_{\infty}
			\le C_{s,\varepsilon}\,N^{-1/2+\varepsilon}
			\quad\text{for all dyadic } N,
			\label{eq:carpet-blocks}
		\end{equation}
		so that $u_{s}$ and $w_{s}$ belong to
		$C^{1/2-\varepsilon}(\mathbb{T}_{2})$ for every $\varepsilon > 0$.
		At every rational $s$ the conclusion fails: $u_{s}$ is a step
		function.
	\end{lemma}
	
	The proof, in Appendix~\ref{app:proofs}, rests on the classical
	Gauss--Weyl bounds for quadratic exponential
	sums~\cite{montgomery1994} along the continued-fraction convergents
	of $s$, which are of balanced size for almost every time. The next
	lemma is the heart of the theory: it identifies exactly where the
	fine-scale energy of the carpet lives. Recall that the block
	$P_{N}u_{s}$ carries total squared mass
	$\asymp N^{-1}$, the frozen Sobolev tail.
	
	\begin{lemma}[Equidistribution of block energy]\label{lem:block-equi}
		Let $I \subseteq (0,1)$ be an interval. For almost every $s$, along
		dyadic $N \to \infty$:
		\begin{equation}
			\begin{gathered}
				N\,\big\Vert P_{N}u_{s}\big\Vert_{L^{2}(I)}^{2}
				\;\longrightarrow\; \frac{2}{\pi^{2}}\,|I|,\\[2pt]
				N\int_{I}\big(\mathrm{Re}\,[e^{i\theta}P_{N}u_{s}]\big)^{2}dz
				\;\longrightarrow\; \frac{1}{\pi^{2}}\,|I|
				\ \text{ uniformly in }\theta,
			\end{gathered}
			\label{eq:block-equi-amp}
		\end{equation}
		while the blocks of the density obey the self-referential law
		\begin{equation}
			N\,\big\Vert P_{2N}^{\mathbb{Z}}\,w_{s}\big\Vert_{L^{2}(I)}^{2}
			\;\longrightarrow\; \frac{2}{\pi^{2}}\int_{I}w_{s}(z)\,dz,
			\label{eq:block-equi-density}
		\end{equation}
		where $P_{2N}^{\mathbb{Z}}$ retains the spatial frequencies
		$|k| \in [2N, 4N)$ of $w_{s}$, in the frequency convention of the
		expansion~\eqref{eq:carpet-seed}: products of odd modes populate
		the even frequencies of $\mathbb{T}_{2}$, so the density is
		$1$-periodic and its spectrum lives on the even sublattice, the
		integer lattice of the standard torus, in contrast to the odd
		amplitude frequencies retained by $P_{N}$.
	\end{lemma}
	
	The amplitude statement says the fine structure of the wave is
	spatially homogeneous: every rank interval holds its Lebesgue share
	of every high block, in every quadrature, with the sharp constant
	$2/\pi^{2}$. The density statement is of a different nature: the
	fine-scale energy of the observable is distributed not according to
	Lebesgue measure but according to \emph{the observable itself}. The
	carpet weaves its fractal detail in proportion to its own local
	mass. This self-similarity in mean is proved in
	Appendix~\ref{app:proofs} by a paraproduct decomposition: the
	high block of $|u_{s}|^{2}$ is, to leading order, twice the real part of
	$\bar u_{s}\times(\text{high block of } u_{s})$, and the amplitude
	equidistribution then delivers the local weight $|u_{s}|^{2}$. The
	variance estimates behind both statements reduce to divisor and
	sum-of-two-squares bounds on the quadratic frequency lattice
	$n^{2} - \tilde n^{2}$, and the constant $2/\pi^{2}$ is the frozen
	mass of the dyadic tail of the seed. One consequence of the
	law~\eqref{eq:block-equi-density} deserves its own statement.
	
	\begin{lemma}[No dark intervals at generic times]\label{lem:dark-intervals}
		For almost every $s$, the kernel $w_{s}$ has positive mass on every
		nonempty open subinterval of $(0,1)$. At rational times this fails:
		entire quantile cells can be dark, as in the quartile
		blackout~\eqref{eq:quartile-blackout}.
	\end{lemma}
	
	\begin{proof}
		Fix an interval with rational endpoints and a smooth $\varphi \ge 0$
		supported in it with $\int\varphi = 1$. Pairing the
		series~\eqref{eq:carpet-seed} with $\varphi$ and folding
		$n \mapsto -n$,
		\begin{equation*}
			\begin{gathered}
				\langle u_{s},\varphi\rangle
				= \sum_{n\ \mathrm{odd},\,n\ge1} b_{n}\,e^{-2\pi in^{2}s}
				= B\big(e^{-2\pi is}\big),\\[2pt]
				B(\zeta) := \sum_{n} b_{n}\,\zeta^{\,n^{2}},
			\end{gathered}
		\end{equation*}
		with $\sum|b_{n}| < \infty$ by smoothness of $\varphi$. The
		function $B$ is analytic and bounded on the unit disk and
		continuous on its closure, and $B \not\equiv 0$ because
		$B(1) = \langle u_{0},\varphi\rangle = \int\varphi = 1$. By the
		boundary uniqueness theorem for bounded analytic
		functions~\cite{duren1970}, the zero set of $B$ on the unit circle
		is Lebesgue-null. If $w_{s}$ vanished on the interval then
		$u_{s}$ would vanish there and $\langle u_{s},\varphi\rangle = 0$;
		hence for each rational interval the exceptional set of times is
		null, and a countable union finishes the proof.
	\end{proof}
	
	Everything is now in place for the main theorem of the section.
	
	\begin{theorem}[The universal quantum carpet]\label{thm:carpet}
		For almost every $\tau$, the following holds for every
		$f \in \mathcal{D}_{+}$ that is continuously differentiable and
		positive on the interior of its support, every compact
		nondegenerate interval $J$ in that interior, and every phase
		$\theta$:
		\begin{enumerate}[label=\normalfont(\roman*), leftmargin=*]
			\item the transported amplitude in every quadrature,
			$\Psi_{\theta,\tau} :=
			\mathrm{Re}\big[e^{i\theta}\,u_{\tau}(F)\big]\sqrt f$, has graph of
			upper box (Minkowski) dimension exactly $\tfrac32$ on $J$;
			\item the density $f_{\tau}$ itself has graph of upper box
			dimension exactly $\tfrac32$ on $J$.
		\end{enumerate}
		At every rational value of the rescaled time $s$, by contrast,
		all these graphs are piecewise continuously differentiable, of
		upper box dimension one.
		The value $\tfrac32$ is independent of the law, the interval, and
		the quadrature: the fractal is woven entirely in rank space and
		printed on every law through its quantile chart.
	\end{theorem}
	
	The proof is given in Appendix~\ref{app:proofs}. Its architecture
	is a pincer. Generic smoothing caps the H\"older exponent at
	$\tfrac12$ from below, hence the dimension at $\tfrac32$ from
	above. Conserved roughness must then surface somewhere, and the
	equidistribution lemma shows it surfaces \emph{everywhere}. On
	every interval, the block mass of the
	display~\eqref{eq:block-equi-amp} together with the sup
	bound~\eqref{eq:carpet-blocks} forces, by interpolation, a lower
	bound on the local $L^{1}$-oscillation of every quadrature at
	scale $N^{-1}$, and an elementary oscillation calculus converts
	this into box dimension $\tfrac32$ from below. For the density
	the same interpolation runs through the self-referential
	law~\eqref{eq:block-equi-density}, whose limit is positive on
	every interval by Lemma~\ref{lem:dark-intervals}. Finally the quantile chart, a
	bi-Lipschitz change of variable on compacts with a positive
	$C^{1}$ multiplier, moves box dimensions without distortion.
	Throughout the paper it is the upper box dimension that is
	controlled: the methods bound oscillation sums along subsequences
	of scales, which determines the limit superior of the box counts
	but not their limit inferior, so the lower box dimension of the
	graphs is not controlled here beyond the trivial bound of one, and
	the Hausdorff dimension of
	density graphs remains open even for step data, as it does
	throughout this literature.
	
	Figure~\ref{fig:boxcount} carries out this box counting on the
	computed sections of Figure~\ref{fig:mosaics}: at the generic time
	the dyadic counts follow the slope $\tfrac32$ of the theorem, while
	at a rational time the mosaic is piecewise constant and the slope
	drops to one.
	
	\begin{figure}[t]
		\centering
		\includegraphics[width=0.62\textwidth]{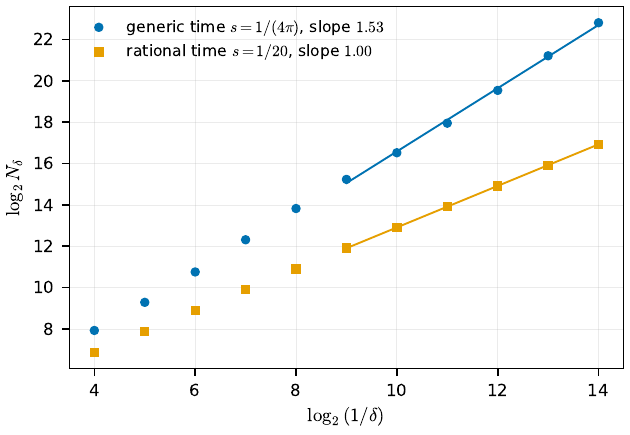}
		\caption{Box-counting regression on sections of the carpet: the
			number $N_{\delta}$ of $\delta$-boxes meeting the graph of $w_{s}$,
			over dyadic scales $\delta = 2^{-4}, \dots, 2^{-14}$, computed from
			the series~\eqref{eq:carpet-seed}. At the generic time
			$s = 1/(4\pi)$ the fitted slope over the five finest octaves is
			$1.530$, matching the dimension $\tfrac32$ of
			Theorem~\ref{thm:carpet}; at the rational time $s = \tfrac{1}{20}$
			the section is a quantile mosaic and the slope is $1.004$.}
		\label{fig:boxcount}
	\end{figure}

	\section{The Strong Law}\label{sec:stronglaw}
	
	The carpet theorem used two properties of the canonical seed: its
	jump, which conserves roughness, and the equidistribution of its
	block energy. This section shows that the first
	property alone suffices, for arbitrary rough seeds and arbitrary
	dispersion, through a mechanism of independent interest in
	mathematical physics: a strong law of large numbers for the critical
	Sobolev mass of the evolved density. Work on the standard torus,
	let $g$ be of bounded variation with jumps $\{J_{j}\}$, let
	$u(x,t) = \sum_{n}\hat g(n)e^{2\pi i(nx - n^{2}t)}$ be its free
	Schr\"odinger evolution, $w = |u|^{2}$ the density, and define the
	critical mass
	\begin{equation}
		Y_{M}(t) := \sum_{m=1}^{M} m\,\big|\hat w(m,t)\big|^{2},
		\label{eq:critical-mass}
	\end{equation}
	a quantity nondecreasing in $M$ and bounded precisely when
	$w(\cdot,t) \in H^{1/2}$. Throughout, the jumps are the
	discontinuities of the periodic extension of $g$. Each $J_{j}$ is
	the difference of the two one-sided limits at a discontinuity,
	which exist everywhere for a function of bounded variation after
	the usual normalization. The endpoint of the period counts once,
	so the square wave has two jumps, each of size two, and square
	summability is automatic, since
	$\sum_{j}|J_{j}| \le V(g) < \infty$. The fractality of the \emph{field} for
	rough data is the theorem of Rodnianski already
	cited~\cite{rodnianski2000}, and the field-level theory has since
	advanced to Schr\"odinger evolutions with
	potentials~\cite{cho2024} and to higher
	dimensions~\cite{erdogan2024}. For the \emph{density}, the
	observable, fractality was known only for step data with jumps at
	rational points, through the Gauss-sum revival arithmetic, and the
	general case was stated as beyond that
	method~\cite{chousionis2015}. Nor does it follow from the
	field-level advances, whose control of $u$ does not transfer to
	$|u|^{2}$. The obstruction is that unitarity controls the field
	and not its modulus, so the conserved roughness could in principle
	hide in the phase and cancel in the intensity.
	
	\begin{theorem}[Strong law for the critical mass]\label{thm:strong-law}
		Let $g$ be real-valued and of bounded variation on the torus with
		at least one jump. Then, with the average taken over a period in
		time,
		\begin{equation}
			\begin{gathered}
				\mathbb{E}_{t}\,Y_{M} = \big(\kappa_{g} + o(1)\big)\log M,
				\qquad
				\operatorname{Var}_{t}(Y_{M}) = O(1),\\[2pt]
				\kappa_{g}
				= \frac{\Vert g\Vert_{L^{2}}^{2}\,\sum_{j}|J_{j}|^{2}}{2\pi^{2}},
			\end{gathered}
			\label{eq:kappa}
		\end{equation}
		and consequently, for almost every $t$,
		\begin{equation}
			\frac{Y_{M}(t)}{\log M}\;\longrightarrow\;\kappa_{g};
			\label{eq:mass-strong-law}
		\end{equation}
		at almost every time the density lies outside $H^{1/2}$, its
		critical mass diverging at the logarithmic rate $\kappa_{g}$,
		independent of the positions and phases of the jumps; moreover,
		for almost every $t$ the density lies in
		$C^{1/2-\varepsilon}$, in no Besov space
		$B^{\sigma}_{1,\infty}$ with $\sigma > \tfrac12$, and its graph has
		upper box dimension exactly $\tfrac32$.
	\end{theorem}
	
	The proof, in Appendix~\ref{app:proofs}, rests on two arithmetic
	pillars and one probabilistic remark. The density coefficients are
	\begin{equation*}
		\hat w(m,t) = \sum_{n}\hat g(n)\,\overline{\hat g(n-m)}\;
		e^{-2\pi im(2n-m)t},
	\end{equation*}
	and within a mode the time frequencies $m(2n-m)$ are
	collision-free, so the time average of $|\hat w(m,t)|^{2}$ is a
	sum of nonnegative terms that no phase arrangement can reduce.
	Wiener's quantitative theorem on the jump content of a function of
	bounded variation~\cite{wiener1924,zygmund2002} then converts the
	weighted sum of these averages into the logarithmic
	law~\eqref{eq:kappa}. Across
	modes, collisions are constrained to the lattice
	$m\Delta = m'\Delta'$ and are sparse in the greatest common
	divisor, so the variance stays bounded while the mean diverges.
	Since $Y_{M}$ is monotone in $M$,
	Chebyshev's inequality along a sparse sequence of scales upgrades
	the mean divergence to the almost-sure
	law~\eqref{eq:mass-strong-law}. A positive-measure set of times
	with bounded critical mass would have to bend a monotone divergent
	mean with bounded mean-square fluctuations, at every scale
	simultaneously, which is impossible.
	The phases keep their freedom at each fixed scale and lose it in
	the aggregate. The average is over the circle of times with
	normalized Lebesgue measure, so the strong law is a deterministic
	second-moment statement in disguise. The density itself, as an
	object of study for rough data, goes back to
	Oskolkov~\cite{oskolkov2006}, who bounded it along rational rays;
	the rational--irrational regularity dichotomy of the field
	originates with Oskolkov~\cite{oskolkov1992} and
	Kapitanski--Rodnianski~\cite{kapitanski1999} before its fractal
	form~\cite{rodnianski2000}. The passage from Sobolev divergence to dimension
	reuses the machinery already in place. The smoothing
	bound~\eqref{eq:carpet-blocks} extends to every seed of bounded
	variation, by writing the block of $u$ as the Stieltjes
	convolution of $dg$ with the block of the canonical Weyl kernel,
	and the oscillation calculus of Lemma~\ref{lem:osc-calculus}
	converts the Besov divergence into the dimension. Figure~\ref{fig:stronglaw}
	illustrates the law for the square-wave datum, whose constant is
	$\kappa_{g} = 4/\pi^{2}$.
	
	\begin{remark}[Complex data]
		Reality of the datum is not essential to
		Theorem~\ref{thm:strong-law}. The collision-freeness of the
		frequencies $m(2n - m)$ within a mode, the Parseval identity
		behind the phase-blind mean, and the two-end localization all
		survive for complex data of bounded variation, and the two ends
		together reconstruct the two-sided Ces\`aro average of Wiener's
		theorem. The theorem therefore holds verbatim, with the same
		constant $\kappa_{g}$, for complex $g$; the only change in the
		proof is that the display splitting the one-sided Wiener sums into
		equal halves is replaced by the sum of the two ends. This widens
		the scope from absorption gratings to the generic case in optics:
		complex transmission functions and phase gratings. Reality is
		essential, by contrast, in the temporal trace of
		Proposition~\ref{prop:dispersion} below, whose mirror family of
		representations at negative spatial frequency requires
		$|\hat g(-n)| = |\hat g(n)|$.
	\end{remark}
	
	\begin{figure}[t]
		\centering
		\includegraphics[width=0.72\textwidth]{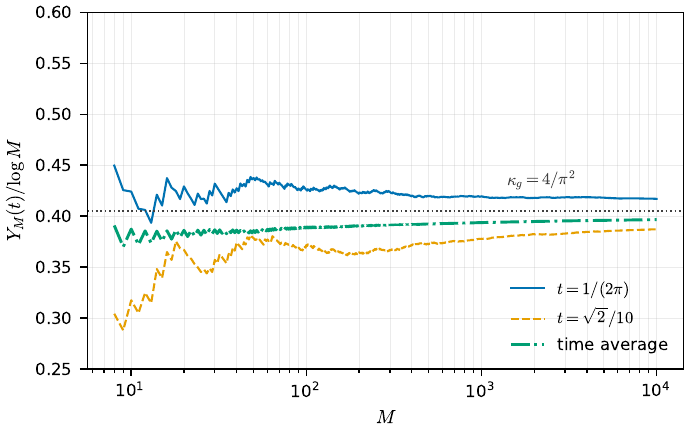}
		\caption{The strong law for the square-wave datum on the torus:
			the normalized critical mass $Y_{M}(t)/\log M$ for two generic
			times and in time average, against the limit
			$\kappa_{g} = 4/\pi^{2} \approx 0.405$ of the
			law~\eqref{eq:mass-strong-law}. The deviation from the limit is of
			order $1/\log M$: the mass itself is $\kappa_{g}\log M + O(1)$ with
			bounded mean-square fluctuations, by the variance bound of the
			theorem, and the division by $\log M$ makes the $O(1)$ term visible
			at this range of $M$.}\label{fig:stronglaw}
	\end{figure}
	
	\begin{corollary}[Universality over seeds and laws]\label{cor:rough-seeds}
		Let the seed $g$ on rank space be of bounded variation with a jump,
		in the sense that its odd extension to the rank circle has a jump,
		and normalized by $\Vert g\Vert_{L^{2}(0,1)} = 1$, so that
		$|e^{-i\tau H_{0}}g|^{2}$ is a curve of kernels by unitarity.
		Then for almost every $\tau$ the deformation
		$\rho_{|e^{-i\tau H_{0}}g|^{2}}[f]$ of any law $f$ whose density is
		continuously differentiable and bounded above and below on its
		support has graph of upper box dimension exactly $\tfrac32$ over
		the whole support, the differentiability being necessary, since a
		rough multiplier $f$ can dominate the oscillation of the product;
		the dimension is global, as in Theorem~\ref{thm:strong-law}, whose
		Besov exclusion is not localized, and it localizes to every
		subinterval for the canonical seed by Theorem~\ref{thm:carpet}.
		For the canonical seed the constant of the
		law~\eqref{eq:mass-strong-law} on the rank circle is
		\begin{equation}
			\kappa_{\mathrm{can}} = \frac{4}{\pi^{2}},
			\label{eq:kappa-canonical}
		\end{equation}
		independent of the law. The intensity carpet of Talbot optics is
		thus fractal, of upper box dimension $\tfrac32$, at almost every
		time for \emph{arbitrary} grating profiles of bounded variation
		with a jump, at arbitrary, possibly irrational, positions. For
		rational step gratings this is the theorem of Chousionis,
		Erdo\u{g}an, and Tzirakis~\cite{chousionis2015}; the general case
		had remained open.
	\end{corollary}
	
	\begin{proposition}[Collision arithmetic of the dispersion]\label{prop:dispersion}
		Let $g$ be real-valued and of bounded variation with a jump.
		\begin{enumerate}[label=\normalfont(\roman*), leftmargin=*]
			\item \textup{(Temporal halving)} For the temporal critical mass
			$Z_{K}(x) := \sum_{k=1}^{K}k^{1/2}\,
			|\hat w^{\,\mathrm{t}}(k;x)|^{2}$ of the time trace at a point,
			one-sided since the density is real and the negative temporal
			frequencies are conjugate, the representations $k = m(2n-m)$ carry
			distinct spatial characters, and, with the average over the circle
			in $x$,
			\begin{equation}
				\mathbb{E}_{x}\,Z_{K}
				= \Big(\frac{\kappa_{g}}{2} + o(1)\Big)\log K.
				\label{eq:temporal-halving}
			\end{equation}
			\item \textup{(Airy doubling)} For the Airy evolution, with
			dispersion $n^{3}$ and a single jump, the within-mode collision
			classes are the pairs $\{n, m-n\}$, their two leading contributions
			are equal and reinforce, and the critical mass
			$Y_{M}^{\mathrm{Airy}}$ of the density obeys
			\begin{equation}
				\mathbb{E}_{t}\,Y_{M}^{\mathrm{Airy}}
				= \big(2\kappa_{g} + o(1)\big)\log M.
				\label{eq:airy-doubling}
			\end{equation}
		\end{enumerate}
	\end{proposition}
	
	\begin{remark}[Almost-sure upgrades]
		We expect both means to upgrade to almost-sure laws, exactly as in
		Theorem~\ref{thm:strong-law}: the temporal
		mass~\eqref{eq:temporal-halving} should satisfy
		$Z_{K}(x)/\log K \to \kappa_{g}/2$ for almost every $x$, placing
		the time trace of the density outside $H^{1/4}$, and the Airy
		mass~\eqref{eq:airy-doubling} should satisfy
		$Y_{M}^{\mathrm{Airy}}(t)/\log M \to 2\kappa_{g}$ for almost every
		$t$. What is missing is in each case a cross-mode variance bound.
		For the trace, the covariance of the masses at temporal
		frequencies $k$ and $k'$ is supported on quadruples of
		representations sharing a spatial frequency, of which there are at
		most $d(k)\,d(k') = O((kk')^{\varepsilon})$ by the divisor bound on
		the representation counts of $k = m(2n - m)$; for the Airy flow,
		completing the square in the within-mode frequency gives
		\begin{equation*}
			4\,Q_{m}(n) = m\,\big(3\Delta^{2} + m^{2}\big),
			\qquad \Delta = 2n - m,
		\end{equation*}
		so cross-mode collisions lie on the conic
		$3m\Delta^{2} - 3m'\Delta'^{2} = m'^{3} - m^{3}$, whose integer
		points should be divisor-bounded in the real quadratic order of
		discriminant $12mm'$. We have not carried out these combinatorial
		estimates in full and state the almost-sure forms as conjectures.
		What is unconditional, beyond the means, is the transfer of
		regularity from field to density along traces: wherever the field
		trace obeys the temporal $C^{1/4-\varepsilon}$ block bound, as it
		does for data of bounded variation~\cite{erdogan2019}, the density
		trace inherits the same exponent, by the elementary inequality
		\begin{equation*}
			\big||u(x,t)|^{2} - |u(x,s)|^{2}\big|
			\;\le\;
			\big(|u(x,t)| + |u(x,s)|\big)\,\big|u(x,t) - u(x,s)\big|
		\end{equation*}
		with the field trace bounded, so the conjectured temporal law
		would place the trace graphs at upper box dimension exactly
		$\tfrac74$, by the same pincer as in Theorem~\ref{thm:strong-law}.
	\end{remark}
	
	\begin{remark}
		Three readings of the constant deserve record. First, it resolves
		the seed-classification question raised by the carpet theorem.
		Every seed of bounded variation with a jump weaves a fractal
		density at almost every time, while for jump-free seeds of bounded
		variation the constant vanishes, $\kappa_{g} = 0$, and the strong
		law asserts no divergence: within the bounded-variation class the
		jump is exactly the threshold for this mechanism, though roughness
		from sources other than jumps is not excluded. Second, the
		constant is a spectrometer. By the law~\eqref{eq:mass-strong-law}
		a single snapshot of the intensity at almost any time determines
		$\Vert g\Vert^{2}\sum_{j}|J_{j}|^{2}$, so the fractal remembers
		the quantitative jump content of a datum it has long since
		dephased; Wiener's statistic, a Fourier-side invariant from 1924,
		reappears as the divergence rate of an observable of quantum
		dynamics. Third, the constant computes the collision arithmetic of
		the dispersion relation: singleton classes give $\kappa_{g}$ for
		the Schr\"odinger flow, quadratic frequency scaling halves it
		along time traces, and the paired classes of the Airy flow double
		it. On product laws the carpets multiply. For tensor seeds the
		density factorizes, and the joint carpet
		$w_{1}(F_{1})\,w_{2}(F_{2})\,f$ of a product law has graph of
		upper box dimension $\tfrac52$ at almost every time, by fiberwise
		box-counting over sections where the second factor is bounded
		below. What replaces this for dependent laws is obstructed by
		infinite hyperplane collision classes and joins the open problems
		of the conclusion.
	\end{remark}
	
	\begin{remark}[Transported Talbot physics]
		Rank space is a quantum box that every probability law carries,
		and the carpet flow is its Talbot dynamics.
		Theorem~\ref{thm:quantile-mosaic} is the fractional Talbot
		revival and Theorem~\ref{thm:carpet} the fractal Talbot profile,
		for the square-wave datum of Berry's quantum
		carpets~\cite{berry1996,berryklein1996,berry2001}, transported by
		the dilation onto arbitrary laws. For real and imaginary parts of
		Schr\"odinger evolutions of bounded-variation data, almost-sure
		dimension $\tfrac32$ is a theorem of
		Rodnianski~\cite{rodnianski2000}, with refinements and nonlinear
		extensions by Erdo\u{g}an and Tzirakis~\cite{erdogan2016}; part
		(i) adapts that skeleton, with the quadrature-uniform and
		localized forms supplied by the equidistribution lemma. For the
		squared modulus, the observable, such statements were previously
		available only for step data with rational
		jumps~\cite{chousionis2015}, the subsequent
		advances~\cite{cho2024,erdogan2024} again concerning the field
		rather than its modulus. Here the gap is closed twice over: by the
		self-referential energy law~\eqref{eq:block-equi-density} for the
		canonical seed, localized to every interval, and by the strong law
		of Theorem~\ref{thm:strong-law} for arbitrary rough seeds, with
		the sharp constant of the display~\eqref{eq:kappa}. The dimension
		$\tfrac32$ itself is the unitary shadow of a boundary phenomenon.
		The seed's Fourier coefficients diverge in $H^{1/2}$ exactly
		because the uniform kernel jumps at the ends of rank space; on the
		unitary axis this endpoint roughness reappears as H\"older
		exponent $\tfrac12$, pinned between conservation and smoothing.
	\end{remark}
	
	\begin{remark}[Experimental reading]
		The prediction is open to measurement, and the two forms of the
		law are not equally robust. A box-count slope needs several
		decades of scale. Even for the exact density sampled at $2^{20}$
		points, the octave-to-octave slopes at a single generic time range
		over $[1.288,\, 1.672]$ around $\tfrac32$; a measured carpet
		offers fewer usable octaves, the detector pitch, the source
		coherence, and the finite ratio of grating period to wavelength
		all truncating the cascade. The spectral form of the law is more
		robust. By Theorem~\ref{thm:strong-law} the critical mass follows
		$Y_{M} = \kappa_{g}\log M + O(1)$ with bounded mean-square
		fluctuations, so fitting the logarithmic growth of the spectral
		mass of a single measured intensity profile, at almost any time,
		returns $\Vert g\Vert^{2}\sum_{j}|J_{j}|^{2}$ directly. Optical
		Talbot carpets have been measured at resolutions approaching what
		this requires~\cite{case2009}; the effect and its applications
		are reviewed in~\cite{wen2013}; matter-wave Talbot and
		Talbot--Lau interferometry supply the same observable for material
		gratings.
	\end{remark}
	

	\section{Conclusion}\label{sec:conclusion}
	
	For the free Schr\"odinger evolution of rough periodic data, the
	fractality of the observable now matches the generality of the
	fractality of the field. For arbitrary real data of bounded
	variation with a jump, the density weaves, at almost every time, a
	fractal of upper box dimension exactly three halves, its critical
	mass diverging at the rate set by Wiener's jump statistic. The
	mechanism is arithmetic, and productive. The same accounting
	halves the constant along time traces and doubles it for the Airy
	flow, so the critical mass of the density is an exactly measurable
	probe of the dispersion relation. Transported by the rank
	calculus, the theory runs on the rank space of every absolutely
	continuous law: quantile mosaics at rational times, a universal
	carpet at almost every other, and an information ledger that
	breathes without producing entropy.
	
	Three problems remain open. Multidimensional data pose hyperplane
	collision classes, for which the phase-blind mean has no present
	replacement. The nonlinear Schr\"odinger density poses the same
	question amid resonant interactions. And the joint carpet of
	dependent laws, beyond the product case of dimension five halves,
	is unresolved. Two refinements lie within reach of the method
	itself. The bounded variance invites a fluctuation theory: the
	distributional limit of the critical mass about its logarithmic
	mean over a random time. The halving and doubling of the constant
	invite a general collision index, one computable factor per
	dispersion relation. The law also invites an experimental
	test, most robustly in its spectral form: the logarithmic growth
	rate of the critical mass of a single measured intensity profile
	returns the jump content of the grating, a more robust observable
	than a fitted box-count slope.
	
	\backmatter
	
	\bmhead{Data availability}
	
	No new data were produced in this work.
	
	\begin{appendices}
		
		\section{Proofs}\label{app:proofs}
		
		This appendix states formally, and proves in detail, the results used
		in the body of the paper whose proofs are routine or computational.

		\begin{lemma}[Oscillation calculus]\label{lem:osc-calculus}
			Let $g$ be continuous and real on an interval, $\delta = 2^{-j}$,
			and, for $\delta \le |I|$, let $\Sigma_{g}(\delta, I)$ denote the
			sum of the oscillations of $g$ over the intersections with $I$ of
			the $\delta$-grid cells meeting the interval $I$.
			Then the number $N_{\delta}(I)$ of $\delta$-boxes meeting the graph
			of $g|_{I}$ satisfies
			\begin{equation}
				\frac{1}{\delta}\,\Sigma_{g}(\delta, I)
				\;\le\; N_{\delta}(I)
				\;\le\; 2\Big(\frac{|I|}{\delta} + 2\Big)
				+ \frac{1}{\delta}\,\Sigma_{g}(\delta, I),
				\label{eq:box-sandwich}
			\end{equation}
			so a bound $g \in C^{\alpha}$ gives upper box dimension at most
			$2 - \alpha$, while
			$\Sigma_{g}(\delta_{i}, I) \ge c\,\delta_{i}^{\,\sigma - 1}$ along
			any sequence $\delta_{i} \to 0$ gives dimension at least
			$2 - \sigma$. Moreover, for $2$-periodic $g$, dyadic $M$, and
			concentric intervals $I' \subset I$ with
			$\operatorname{dist}(I', \complement I) \ge \delta_{0} > 0$,
			\begin{equation}
				\big\Vert P_{M}\,g\big\Vert_{L^{1}(I')}
				\;\le\; C\,\delta\,\Sigma_{g}(\delta, I)
				+ C_{A}\,\Vert g\Vert_{\infty}\,(M\delta_{0})^{1-A},
				\qquad \delta = 1/M,
				\label{eq:l1-oscillation}
			\end{equation}
			for every $A$, where $P_{M}$ is any smooth dyadic block at
			frequency scale $M$, in the amplitude or the density convention.
		\end{lemma}
		
		\begin{proof}
			Over a grid cell the graph meets at least
			$\operatorname{osc}/\delta$ and at most
			$2 + \operatorname{osc}/\delta$ vertical boxes, which is the
			sandwich~\eqref{eq:box-sandwich}, a standard box-counting
			argument~\cite{falconer2014}, and the two dimension statements
			follow by taking logarithms. For the
			inequality~\eqref{eq:l1-oscillation}, realize the block as a
			convolution $P_{M}g = g * K_{M}$ with a kernel of integral zero and
			rapid decay, $|K_{M}(y)| \le C_{A}M(1 + M|y|)^{-A}$. Then
			\begin{equation*}
				|P_{M}g(x)| \le \int |K_{M}(y)|\,\big|g(x-y) - g(x)\big|\,dy,
			\end{equation*}
			and integrating over $x \in I'$: a shift by $|y| \le \delta$
			changes $g$ by at most the oscillations of the two grid cells
			involved, so
			$\Vert g(\cdot - y) - g\Vert_{L^{1}(I')} \le
			2\delta\,\Sigma_{g}(\delta, I)$; a shift by
			$\delta \le |y| \le \delta_{0}$ is a chain of at most
			$2|y|/\delta$ such steps inside $I$, so the corresponding integral
			is at most
			$4\,\Sigma_{g}(\delta, I)\int |y|\,|K_{M}(y)|\,dy
			\le C\delta\,\Sigma_{g}(\delta, I)$; and the far tail
			$|y| > \delta_{0}$ contributes at most
			$2\Vert g\Vert_{\infty}\int_{|y|>\delta_{0}}|K_{M}|
			\le C_{A}\Vert g\Vert_{\infty}(M\delta_{0})^{1-A}$.
		\end{proof}
		
		\begin{lemma}[Chart transport]\label{lem:chart-transport}
			Let $f \in \mathcal{D}_{+}$ be continuously differentiable and
			positive on the interior of its support, let $J$ be a compact
			interval there, and let $g$ be continuous near $F(J)$. For any
			positive $h \in C^{1}(J)$, the graphs of $g|_{F(J)}$ and of
			$(g\circ F)\,h|_{J}$ have the same upper box dimension whenever
			either exceeds one.
		\end{lemma}
		
		\begin{proof}
			On $J$ the map $F$ is a bi-Lipschitz $C^{1}$ diffeomorphism onto
			$F(J)$, so grid oscillations of $g\circ F$ at scale $\delta$ are
			sandwiched between oscillation sums of $g$ at the comparable scales
			$c\delta$ and $C\delta$, changing the counts of the
			sandwich~\eqref{eq:box-sandwich} by bounded factors only. For the
			multiplier, with $h \ge c > 0$ and
			$\operatorname{osc}_{I}(h) \le C\delta$,
			\begin{equation*}
				\operatorname{osc}_{I}\big((g\circ F)h\big)
				\;\ge\; c\,\operatorname{osc}_{I}(g\circ F)
				- \Vert g\Vert_{\infty}\,C\delta,
			\end{equation*}
			and symmetrically above, so the oscillation sums differ by at most
			a bounded factor plus $O(1)$, which box counts of order
			$\delta^{-\beta}$ with $\beta > 1$ do not feel.
		\end{proof}
		
		\renewcommand{\restatename}{Lemma~\ref{lem:carpet-smoothing}}%
		\begin{restatement}[Generic smoothing]
			For almost every $s$: $u_{s}$ is continuous, and for every
			$\varepsilon > 0$ there is $C_{s,\varepsilon}$ with
			\begin{equation*}
				\Vert P_{N}u_{s}\Vert_{\infty}
				\le C_{s,\varepsilon}\,N^{-1/2+\varepsilon}
				\quad\text{for all dyadic } N,
			\end{equation*}
			so that $u_{s}$ and $w_{s}$ belong to
			$C^{1/2-\varepsilon}(\mathbb{T}_{2})$ for every $\varepsilon > 0$.
			At every rational $s$ the conclusion fails: $u_{s}$ is a step
			function.
		\end{restatement}
		
		\begin{proof}
			Fix $\varepsilon > 0$. For a.e.\ $s$ the following holds for all
			large dyadic $N$: there is a reduced fraction $a/q$ with
			$|s - a/q| \le 1/(qN)$ and $N^{1-\varepsilon} \le q \le N$. Indeed
			Dirichlet's theorem always provides $q \le N$, while the set of
			$s$ admitting some $q \le N^{1-\varepsilon}$ has measure at most
			$\sum_{q \le N^{1-\varepsilon}} q \cdot 2/(qN)
			\le 2N^{-\varepsilon}$, which is summable along dyadic $N$, so by
			the Borel--Cantelli lemma small denominators eventually never
			occur. On this event the Gauss--Weyl
			estimate~\cite{montgomery1994} gives, uniformly in the linear
			phase $x$,
			\begin{equation*}
				\sup_{x}\Big|\sum_{l=1}^{L}
				e^{2\pi i(l^{2}s' + lx)}\Big|
				\;\le\; C\Big(\frac{L}{\sqrt q} + \sqrt{q\log q}\Big)
				\;\le\; C\,L^{1/2 + \varepsilon}
				\qquad (L \asymp N,\ \text{the only range used}),
			\end{equation*}
			for $s'$ in the same approximation class; the substitution
			$n = 2l+1$ writes the odd-frequency block of the
			series~\eqref{eq:carpet-seed} as such a sum with time parameter
			$4s$, and the map $s \mapsto 4s$ modulo one preserves null sets.
			Abel summation against the monotone coefficients
			$2/(\pi n) \asymp N^{-1}$ on the block then yields the
			bound~\eqref{eq:carpet-blocks}. Summing the blocks shows $u_{s}$
			is a uniform limit of continuous functions; the geometric block
			decay is the H\"older--Besov characterization of
			$C^{1/2-\varepsilon}$~\cite{katznelson2004}, and
			$w_{s} = |u_{s}|^{2}$ is a product of bounded
			$C^{1/2-\varepsilon}$ functions. At rational $s$ the
			revival~\eqref{eq:carpet-revival} is a step function.
		\end{proof}
		
		\renewcommand{\restatename}{Lemma~\ref{lem:block-equi}}%
		\begin{restatement}[Equidistribution of block energy]
			Let $I \subseteq (0,1)$ be an interval. For almost every $s$, along
			dyadic $N \to \infty$:
			\begin{equation*}
				\begin{gathered}
					N\,\big\Vert P_{N}u_{s}\big\Vert_{L^{2}(I)}^{2}
					\;\longrightarrow\; \frac{2}{\pi^{2}}\,|I|,\\[2pt]
					N\int_{I}\big(\mathrm{Re}\,[e^{i\theta}P_{N}u_{s}]\big)^{2}dz
					\;\longrightarrow\; \frac{1}{\pi^{2}}\,|I|
					\ \text{ uniformly in }\theta,
				\end{gathered}
			\end{equation*}
			while the blocks of the density obey the self-referential law
			\begin{equation*}
				N\,\big\Vert P_{2N}^{\mathbb{Z}}\,w_{s}\big\Vert_{L^{2}(I)}^{2}
				\;\longrightarrow\; \frac{2}{\pi^{2}}\int_{I}w_{s}(z)\,dz,
			\end{equation*}
			where $P_{2N}^{\mathbb{Z}}$ retains the spatial frequencies
			$|k| \in [2N, 4N)$ of $w_{s}$, in the frequency convention of the
			expansion~\eqref{eq:carpet-seed}: products of odd modes populate
			the even frequencies of $\mathbb{T}_{2}$, so the density is
			$1$-periodic and its spectrum lives on the even sublattice, the
			integer lattice of the standard torus, in contrast to the odd
			amplitude frequencies retained by $P_{N}$.
		\end{restatement}
		
		\begin{proof}
			Throughout, $P_{N}$ may be taken as the sharp block or as a
			smooth dyadic block. Every estimate below uses only the
			boundedness of the multiplier weights, and for a smooth block the
			limiting constants are multiplied by the positive weight integral,
			which is all the dimension argument uses. The localized
			oscillation inequality~\eqref{eq:l1-oscillation} is applied to
			the smooth version. It suffices to prove each statement for intervals with
			rational endpoints and every tolerance
			$\varepsilon = 1/\ell$, the general case following by approximation
			from inside and outside, and the almost-sure sets being intersected
			over the countable family at the end. With the coefficients
			$c_{n}$ of the seed~\eqref{eq:carpet-seed} and a bounded test
			function $\varphi$, write
			\begin{equation}
				X_{N}^{\varphi}(s)
				:= \big\langle \varphi,\ |P_{N}u_{s}|^{2}\big\rangle
				= \sum_{n, \tilde n}
				c_{n}\overline{c_{\tilde n}}\,
				e^{-2\pi i(n^{2} - \tilde n^{2})s}\,
				\overline{\hat\varphi}(n - \tilde n),
				\label{eq:test-expansion}
			\end{equation}
			the sums over odd frequencies in the block, with
			$\hat\varphi$ the coefficients on $\mathbb{T}_{2}$.
			
			First, the mean. Averaging over the period retains
			$\tilde n = \pm n$: the diagonal gives
			\begin{equation*}
				\Big(\sum_{\mathrm{block}}|c_{n}|^{2}\Big)\int\varphi\,dz
				= \frac{2}{\pi^{2}N}\,\big(1 + O(N^{-1})\big)\int\varphi\,dz,
			\end{equation*}
			by the block mass of the seed coefficients, and the antidiagonal
			$\tilde n = -n$ contributes at most
			$\sum_{\mathrm{block}}(4/\pi^{2}n^{2})\,|\hat\varphi(2n)|
			= O(N^{-2})$ for $\varphi$ of bounded variation.
			
			Second, the variance. The fluctuating part of the
			expansion~\eqref{eq:test-expansion} is supported on the time
			frequencies $\mu = n^{2} - \tilde n^{2} = k(n + \tilde n)$ with
			$k = n - \tilde n \neq 0$, and for fixed $\mu$ and $k \mid \mu$ the
			pair $(n, \tilde n)$ is determined. Hence, by Parseval in $s$,
			\begin{equation*}
				\operatorname{Var}_{s} X_{N}^{\varphi}
				\le \frac{C}{N^{4}}\sum_{\mu \neq 0}
				\Big(\sum_{k \mid \mu}|\hat\varphi(k)|\Big)^{2}
				= \frac{C}{N^{4}}\sum_{k, k'}
				|\hat\varphi(k)||\hat\varphi(k')|\,
				\#\{\mu\},
			\end{equation*}
			where $\#\{\mu\}$ counts the frequencies divisible by both $k$ and
			$k'$ that the block realizes; since the count is symmetric in $k$
			and $k'$, the larger denominator may be used, giving
			\begin{equation*}
				\#\{\mu\} \;\le\; \frac{2N\gcd(k,k')}{\max(|k|,|k'|)} + 1.
			\end{equation*}
			The test functions used below are
			$\varphi = \chi_{I}e^{i\pi jz}$ with $|j| \le 2\sqrt N$, whose
			coefficients obey
			$|\hat\varphi(k)| \le \min(1, C/\langle k - j\rangle)$. The
			resulting bound is uniform in $j$: the equal-frequency terms
			$k = k'$ carry $\gcd = |k|$ and contribute
			$CN^{-4}\cdot 2N\sum_{k}\langle k - j\rangle^{-2} \le CN^{-3}$,
			while for $k \neq k'$ the divisor $\gcd(k, k') \le |k - k'|$ and
			the shifted peaks give, after summing the harmonic tails, at most
			$CN^{-3}\log^{3}N$ regardless of the location of $j$; hence
			\begin{equation}
				\operatorname{Var}_{s} X_{N}^{\varphi}
				\;\le\; C\,N^{-3}\log^{3} N
				\qquad\text{uniformly in } |j| \le 2\sqrt N.
				\label{eq:test-variance}
			\end{equation}
			
			Third, the quadrature term. With
			$Y_{N}(s) := \int_{I}(P_{N}u_{s})^{2}\,dz
			= \sum c_{n}c_{\tilde n}e^{-2\pi i(n^{2} + \tilde n^{2})s}
			\hat\chi_{I}$-type coefficients, every time frequency
			$n^{2} + \tilde n^{2}$ is positive, so $\mathbb{E}_{s}Y_{N} = 0$;
			the number of representations of $\mu$ as an ordered sum of two odd
			squares is at most a divisor function, $O(\mu^{\varepsilon})$, so
			the same Parseval argument gives
			$\operatorname{Var}_{s}Y_{N} \le C_{\varepsilon}N^{-3+\varepsilon}$.
			
			Fourth, the assembly. Fix $\varepsilon > 0$. By the
			uniform bound~\eqref{eq:test-variance}, which applies verbatim to
			the frequency-window blocks used in the fifth step, and by
			Chebyshev's inequality,
			the probability that some test in the family
			$\{\chi_{I}e^{i\pi jz} : |j| \le 2\sqrt N\}$, evaluated on the
			dyadic block or on the window block of the fifth step, deviates
			from its mean
			by more than $\varepsilon N^{-1}\log^{-2}N$, or that
			$|Y_{N}| > \varepsilon N^{-1}$, is at most
			$4\sqrt N\cdot C\varepsilon^{-2}N^{-1}\log^{7}N
			= C_{\varepsilon}N^{-1/2}\log^{7}N$, summable along dyadic $N$; by
			the Borel--Cantelli lemma, almost surely all these deviations fail
			for all large dyadic $N$. The case $j = 0$ gives the first
			limit of the display~\eqref{eq:block-equi-amp}; adding
			$\tfrac12\operatorname{Re}\,[e^{2i\theta}Y_{N}]$ to
			$\tfrac12 X_{N}^{\chi_{I}}$ gives the quadrature limit, uniformly
			in $\theta$ since $|{\operatorname{Re}\,e^{2i\theta}Y_{N}}| \le |Y_{N}|$.
			
			Fifth, the density and the self-referential law. Set
			$K := \lceil\sqrt N\,\rceil$, $v := P_{<K}u_{s}$, and let $h$ be
			the part of $u_{s}$ with frequencies of absolute value in the
			window $[2N - K,\, 4N + K)$. Splitting each pair
			$(n,\, n - k)$ with $|k| \in [2N, 4N)$ according to which index is
			small, the block of the density decomposes as
			\begin{equation}
				P_{2N}^{\mathbb{Z}}\,w_{s}
				= \Pi_{N}\big(\bar v\,h + v\,\bar h\big) + R_{N},
				\label{eq:paraproduct}
			\end{equation}
			where $\Pi_{N}$ restricts to the spatial window and the remainder
			$R_{N}$ collects the pairs with both indices at least $K$. The
			remainder is negligible: its time frequencies $k(2n - k)$ are
			distinct over $n$ for fixed $k$, so
			\begin{equation*}
				\mathbb{E}_{s}\Vert R_{N}\Vert_{L^{2}(0,1)}^{2}
				= \sum_{k}\sum_{n}|c_{n}|^{2}|c_{n-k}|^{2}
				\,\mathbf{1}\{\text{both} \ge K\}
				\;\le\; \frac{C}{KN},
			\end{equation*}
			and Markov's inequality with $K = \sqrt N$ makes
			$\Vert R_{N}\Vert_{2}^{2} \le \varepsilon N^{-1}$ eventually,
			almost surely, along dyadic $N$. For the main term, expand
			$|\Pi_{N}(\bar vh + v\bar h)|^{2}
			= 2|v|^{2}|h|^{2} + 2\operatorname{Re}\,[\bar v^{2}h^{2}]$ up to
			window-edge corrections of deterministic $L^{2}$ mass
			$O(K\log^{2}K/N^{2})$, still $o(N^{-1})$ at $K = \sqrt N$. Writing
			$|v|^{2} = \sum_{|j| < 2K}V_{j}e^{i\pi jz}$ with
			$\sum_{j}|V_{j}| \le \big(\sum_{|n|<K}|c_{n}|\big)^{2}
			\le C\log^{2}N$,
			the first term integrates against $\chi_{I}$ to
			\begin{equation*}
				2\sum_{j}V_{j}\,
				\big\langle \chi_{I}e^{-i\pi jz},\ |h|^{2}\big\rangle
				\;=\; \frac{2}{\pi^{2}N}\int_{I}|v|^{2}\,dz
				\;+\; o(N^{-1}),
			\end{equation*}
			by the already-established equidistribution of the window blocks
			of $u$ against the test family, with the error uniform over the
			$O(\sqrt N)$ tests by the union bound above; the window mass
			replacing the dyadic block mass is
			\begin{equation*}
				\sum_{\substack{|n|\,\in\,[2N-K,\,4N+K)\\ n\ \mathrm{odd}}}
				|c_{n}|^{2}
				= \frac{8}{\pi^{2}}
				\sum_{\substack{2N-K\,\le\, n\,<\,4N+K\\ n\ \mathrm{odd}}}
				\frac{1}{n^{2}}
				= \frac{1}{\pi^{2}N}\,\big(1 + o(1)\big)
			\end{equation*}
			at $K = \lceil\sqrt N\,\rceil$, half the mass $2/(\pi^{2}N)$ of a
			dyadic block per unit of the normalization. Since
			$v \to u_{s}$ in $L^{2}$, the integral
			$\int_{I}|v|^{2}$ converges to $\int_{I}w_{s}$. The oscillatory
			term $2\operatorname{Re}\int_{I}\bar v^{2}h^{2}$ has zero mean, its
			time frequencies being sums of two squares of window size shifted
			by $O(K^{2})$, and the divisor bound gives variance
			$O(N^{-3+\varepsilon}\log^{4}N)$; Chebyshev and Borel--Cantelli
			finish as before. Collecting the pieces and letting
			$\varepsilon \downarrow 0$ along a sequence proves the
			law~\eqref{eq:block-equi-density}.
		\end{proof}
		
		\renewcommand{\restatename}{Theorem~\ref{thm:carpet}}%
		\begin{restatement}[The universal quantum carpet]
			For almost every $\tau$, the following holds for every
			$f \in \mathcal{D}_{+}$ that is continuously differentiable and
			positive on the interior of its support, every compact
			nondegenerate interval $J$ in that interior, and every phase
			$\theta$:
			\begin{enumerate}[label=\normalfont(\roman*), leftmargin=*]
				\item the transported amplitude in every quadrature,
				$\Psi_{\theta,\tau} :=
				\mathrm{Re}\big[e^{i\theta}\,u_{\tau}(F)\big]\sqrt f$, has graph of
				upper box (Minkowski) dimension exactly $\tfrac32$ on $J$;
				\item the density $f_{\tau}$ itself has graph of upper box
				dimension exactly $\tfrac32$ on $J$.
			\end{enumerate}
			At every rational value of the rescaled time $s$, by contrast,
			all these graphs are piecewise continuously differentiable, of
			upper box dimension one.
			The value $\tfrac32$ is independent of the law, the interval, and
			the quadrature: the fractal is woven entirely in rank space and
			printed on every law through its quantile chart.
		\end{restatement}
		
		\begin{proof}
			Fix the full-measure set of times on which
			Lemma~\ref{lem:carpet-smoothing}, Lemma~\ref{lem:block-equi} for
			every rational interval, and Lemma~\ref{lem:dark-intervals} all
			hold, and fix such a time $s$, a law $f$ in the stated class, a
			compact nondegenerate $J$, and a phase $\theta$.
			
			Upper bounds. By Lemma~\ref{lem:carpet-smoothing} the functions
			$g_{\theta} := \operatorname{Re}\,[e^{i\theta}u_{s}]$ and $w_{s}$
			are $C^{1/2 - \varepsilon}$; composition with the bi-Lipschitz
			chart and multiplication by a $C^{1}$ factor preserve the class on
			$J$, so by Lemma~\ref{lem:osc-calculus} both graphs have dimension
			at most $\tfrac32 + \varepsilon$ for every $\varepsilon$.
			
			Lower bounds. Choose rational intervals
			$I' \subset\subset I \subset\subset \operatorname{int}F(J)$, with
			$I'$ concentric in $I$. For the amplitude, the quadrature limit of
			the display~\eqref{eq:block-equi-amp} on $I'$ and the sup
			bound~\eqref{eq:carpet-blocks} give, for all large dyadic $N$,
			\begin{equation*}
				\big\Vert P_{N}g_{\theta}\big\Vert_{L^{1}(I')}
				\;\ge\;
				\frac{\int_{I'}(P_{N}g_{\theta})^{2}}
				{\Vert P_{N}g_{\theta}\Vert_{\infty}}
				\;\ge\;
				\frac{c_{0}\,|I'|/N}
				{C_{s,\varepsilon}N^{-1/2+\varepsilon}}
				\;\ge\; c\,N^{-1/2 - \varepsilon},
			\end{equation*}
			with $c_{0} > 0$ the quadrature-limit constant of the smooth block,
			of which only positivity is used; the localized block
			inequality~\eqref{eq:l1-oscillation}, read from right to left with
			$\delta = 1/N$, forces
			$\Sigma_{g_{\theta}}(\delta, I) \ge c\,\delta^{-1/2 + 2\varepsilon}$
			for all small dyadic $\delta$; by
			Lemma~\ref{lem:osc-calculus} the graph of $g_{\theta}$ on $I$ has
			dimension at least $\tfrac32 - 2\varepsilon$. For the density the
			same chain runs through the self-referential
			law~\eqref{eq:block-equi-density}, whose limit
			$(2/\pi^{2})\int_{I'}w_{s}\,dz$ is strictly positive by
			Lemma~\ref{lem:dark-intervals}, giving
			$\Sigma_{w_{s}}(\delta, I) \ge c_{I}(s)\,\delta^{-1/2+2\varepsilon}$
			and dimension at least $\tfrac32 - 2\varepsilon$ for the graph of
			$w_{s}$ on $I$.
			
			Transport. On $J$ the transported amplitude is
			$(g_{\theta}\circ F)\sqrt f$ and the density is
			$(w_{s}\circ F)\,f$; the quantile image of the compact
			$\overline{Q(I)} \subset J$ is $\overline I$, so
			Lemma~\ref{lem:chart-transport} carries both lower bounds, and the
			upper bounds, through the chart. Letting
			$\varepsilon \downarrow 0$ along a sequence gives dimension exactly
			$\tfrac32$ in parts (i) and (ii). The rational-time statement is
			immediate from the mosaic theorem: at rational $s$ the kernel is a
			quantile histogram, so the transported graphs are piecewise
			$C^{1}$ over the cell partition, of upper box dimension one.
		\end{proof}
		
		\renewcommand{\restatename}{Theorem~\ref{thm:strong-law}}%
		\begin{restatement}[Strong law for the critical mass]
			Let $g$ be real-valued and of bounded variation on the torus with
			at least one jump. Then, with the average taken over a period in
			time,
			\begin{equation*}
				\begin{gathered}
					\mathbb{E}_{t}\,Y_{M} = \big(\kappa_{g} + o(1)\big)\log M,
					\qquad
					\operatorname{Var}_{t}(Y_{M}) = O(1),\\[2pt]
					\kappa_{g}
					= \frac{\Vert g\Vert_{L^{2}}^{2}\,\sum_{j}|J_{j}|^{2}}{2\pi^{2}},
				\end{gathered}
			\end{equation*}
			and consequently, for almost every $t$,
			\begin{equation*}
				\frac{Y_{M}(t)}{\log M}\;\longrightarrow\;\kappa_{g};
			\end{equation*}
			at almost every time the density lies outside $H^{1/2}$, its
			critical mass diverging at the logarithmic rate $\kappa_{g}$,
			independent of the positions and phases of the jumps; moreover,
			for almost every $t$ the density lies in
			$C^{1/2-\varepsilon}$, in no Besov space
			$B^{\sigma}_{1,\infty}$ with $\sigma > \tfrac12$, and its graph has
			upper box dimension exactly $\tfrac32$.
		\end{restatement}
		
		\begin{proof}
			Write $y_{n} := \hat g(n)\overline{\hat g(n-m)}$, so that
			$\hat w(m,t) = \sum_{n}y_{n}e^{-2\pi im(2n-m)t}$, with
			$|\hat g(n)| \le C/\langle n\rangle$ for bounded variation. The
			frequencies $m(2n-m)$ are strictly increasing in $n$ for
			$m \ge 1$, so Parseval in time gives the exact, phase-blind mean
			\begin{equation*}
				\sigma_{m}^{2} := \mathbb{E}_{t}\big|\hat w(m,t)\big|^{2}
				= \sum_{n}|\hat g(n)|^{2}\,|\hat g(n-m)|^{2}.
			\end{equation*}
			The mean of the mass~\eqref{eq:critical-mass} is
			$\sum_{m\le M}m\,\sigma_{m}^{2}$, and as $m$ grows the sum
			defining $\sigma_{m}^{2}$ localizes at its two ends. Split the
			$n$-sum into the ranges $2|n| \le m$ and $2|n - m| \le m$, which
			overlap in at most the single term $2n = m$, of size $O(m^{-4})$,
			and the remainder: on the first end $|\hat g(n-m)|^{2}$ is of size
			$m^{-2}$ and $|\hat g(n)|^{2}$ carries the Wiener content, on the
			second end the roles are exchanged, and the remainder, where both
			$|n|$ and $|n - m|$ exceed $m/2$, contributes $O(m^{-3})$ per mode
			and $O(1)$ in total. Since
			$\widehat{dg}(n) = 2\pi in\,\hat g(n)$, Wiener's
			theorem~\cite{wiener1924,zygmund2002} gives the Ces\`aro asymptotics
			\begin{equation*}
				\frac{1}{2N}\sum_{|n|\le N}\big|2\pi n\,\hat g(n)\big|^{2}
				\;\longrightarrow\;\sum_{j}|J_{j}|^{2},
			\end{equation*}
			and by Abel summation against the weight $1/k$,
			\begin{equation*}
				\sum_{1 \le k \le M}k\,|\hat g(-k)|^{2}
				+ \sum_{1 \le k \le M}k\,|\hat g(k)|^{2}
				= \frac{\sum_{j}|J_{j}|^{2}}{2\pi^{2}}\,\log M\,(1 + o(1)),
			\end{equation*}
			each of the two sums carrying half, since $g$ is real and
			$|\hat g(-k)| = |\hat g(k)|$. The insertion of these asymptotics
			into the two ends involves an exchange of summations that we
			display rather than imply, because the unlocalized inner sums
			$\sum_{m \le M}m\,|\hat g(n-m)|^{2}$ are \emph{not} uniformly
			$O(\log M)$ in $n$: near $m = n$ they can be of order
			$\langle n\rangle$, and it is the end localization that removes
			this growth. On the first end the exchange reads
			\begin{equation}
				\begin{gathered}
					\sum_{m\le M}\,m\!\!\sum_{2|n| \le m}\!\!
					|\hat g(n)|^{2}\,|\hat g(n-m)|^{2}
					= \sum_{n}|\hat g(n)|^{2}\;T_{M}(n),\\[2pt]
					T_{M}(n) := \!\!\!\sum_{\substack{m \le M\\ m \ge 2|n| \vee 1}}
					\!\!\! m\,|\hat g(n-m)|^{2},
				\end{gathered}
				\label{eq:mean-exchange}
			\end{equation}
			and the substitution $k = m - n$, which keeps $k \ge |n|$
			throughout the range of $T_{M}(n)$, exactly so at $n \ge 0$ and
			with room to spare for $n < 0$, where the range begins at $3|n|$,
			gives
			\begin{equation*}
				T_{M}(n)
				= \sum_{k}k\,|\hat g(-k)|^{2} + n\!\sum_{k \ge |n|}\!|\hat g(-k)|^{2}
				= \sum_{1 \le k \le M}k\,|\hat g(-k)|^{2}
				+ O\big(\log\langle n\rangle + 1\big),
			\end{equation*}
			where the middle term is $O(1)$ by
			$\langle n\rangle\sum_{k \ge |n|}Ck^{-2} \le C$ and the endpoint
			corrections of the $k$-range cost harmonic tails, of size
			$O(\log\langle n\rangle)$ at the bottom and $O(1)$ at the top. In
			particular $T_{M}(n) \le C\log M$ uniformly on the admissible range
			$2|n| \le m \le M$, while at each fixed $n$ the Abel-summed Wiener
			asymptotics give
			$T_{M}(n)/\log M \to \sum_{j}|J_{j}|^{2}/4\pi^{2}$; dominated
			convergence in $n$, against the summable weights
			$|\hat g(n)|^{2}$ with $\sum_{n}|\hat g(n)|^{2}
			= \Vert g\Vert_{2}^{2}$, therefore yields
			\begin{equation*}
				\sum_{n}|\hat g(n)|^{2}\,T_{M}(n)
				= \frac{\Vert g\Vert_{2}^{2}\sum_{j}|J_{j}|^{2}}{4\pi^{2}}
				\,\log M + o(\log M).
			\end{equation*}
			The second end, parametrized by $j = n - m$ with $2|j| \le m$, is
			symmetric, its inner sums running over the coefficients at positive
			frequencies, and contributes the same amount. Hence
			\begin{equation*}
				\sum_{m \le M}m\,\sigma_{m}^{2}
				= 2\cdot\frac{\Vert g\Vert_{2}^{2}\sum_{j}|J_{j}|^{2}}{4\pi^{2}}
				\,\log M + o(\log M)
				= \big(\kappa_{g} + o(1)\big)\log M,
			\end{equation*}
			the factor two counting the two ends; the error is $o(\log M)$ and
			not better in general, Wiener's theorem carrying no rate for
			countably many jumps or a continuous singular part, and it improves
			to $O(1)$ for finitely many jumps with an absolutely continuous
			remainder. This is the mean of
			the display~\eqref{eq:kappa}, and positivity of every term is what
			makes it impervious to the jump positions and to the phases.
			
			For the variance, expand the fluctuation
			$|\hat w(m,t)|^{2} - \sigma_{m}^{2}
			= \sum_{\Delta\neq0}R_{m}(\Delta)\,e^{-4\pi im\Delta t}$ with
			$R_{m}(\Delta) = \sum_{n}y_{n}\bar y_{n-\Delta}$, the
			autocorrelation of the coefficient sequence. Three bounds on the
			autocorrelation are needed, all from the localization of the sum at
			the four points $0, m, \Delta, \Delta + m$ where a factor peaks.
			Generically,
			\begin{equation*}
				|R_{m}(\Delta)| \;\le\;
				\frac{C\log^{2}(2 + |m\Delta|)}
				{\langle m\rangle\langle\Delta\rangle
					\langle\max(|m|,|\Delta|)\rangle}
				\qquad
				\big(\,\big||\Delta| - m\big| > \tfrac{m}{2}\,\big);
			\end{equation*}
			in the resonant windows two of the four points merge, and
			localizing at $n = m + j$ gives the refined bound
			\begin{equation*}
				\big|R_{m}(\pm m + r)\big|
				\;\le\; \frac{C}{m^{2}}\sum_{j}
				\frac{1}{\langle j\rangle\langle j - r\rangle}
				\;\le\; \frac{C\log(2 + |r|)}{m^{2}\,\langle r\rangle},
			\end{equation*}
			which decays in the offset $r$; and for the single-mode variance
			the sharpest route is through the fourth power of the trigonometric
			polynomial $Y_{m}(\theta) = \sum_{n}y_{n}e^{2\pi in\theta}$,
			\begin{equation*}
				\operatorname{Var}\big(|\hat w(m)|^{2}\big)
				= \sum_{\Delta \neq 0}|R_{m}(\Delta)|^{2}
				\le \int_{0}^{1}|Y_{m}|^{4}
				\le \Vert y\Vert_{1}^{2}\,\Vert y\Vert_{2}^{2}
				\le \frac{C\log^{2}m}{m^{2}}\,\sigma_{m}^{2}
				\le \frac{C\log^{2}m}{m^{4}},
			\end{equation*}
			by $\Vert y\Vert_{1} \le C\log m/m$. Now decompose the variance,
			\begin{equation*}
				\operatorname{Var}_{t}(Y_{M})
				= \sum_{m,m'\le M}mm'\,\operatorname{Cov}
				\big(|\hat w(m)|^{2},\, |\hat w(m')|^{2}\big),
			\end{equation*}
			into the diagonal, the separated pairs, and the near-diagonal band.
			The diagonal is
			\begin{equation*}
				\sum_{m\le M}m^{2}\operatorname{Var}\big(|\hat w(m)|^{2}\big)
				\le C\sum_{m}\frac{\log^{2}m}{m^{2}} = O(1).
			\end{equation*}
			For $m \neq m'$ the
			covariance is supported on the collisions
			$m\Delta = m'\Delta'$: writing $d = \gcd(m, m')$, $m = d\mu$,
			$m' = d\mu'$ with $\mu, \mu'$ coprime, they are
			$\Delta = \mu'\tau$, $\Delta' = \mu\tau$, $\tau \neq 0$. If
			neither factor is resonant, the generic bound gives, per pair,
			$\operatorname{Cov} \le C\log^{4}M\,d^{2}/(m^{3}m'^{3})$ after
			summing the geometric $\tau$-series, and
			\begin{equation*}
				\sum_{m\le M}\sum_{m'\le M}mm'\cdot
				\frac{d^{2}}{m^{3}m'^{3}}
				\;\le\;
				\sum_{m\le M}\frac{1}{m^{2}}\sum_{d \mid m}
				d^{2}\sum_{\mu'}\frac{1}{(d\mu')^{2}}
				\;\le\; C\sum_{m\le M}\frac{\sigma_{0}(m)}{m^{2}}
				= O(1),
			\end{equation*}
			with $\sigma_{0}$ the divisor count, so the separated and mixed
			regimes contribute $O(\log^{4}M)\cdot O(1)$-summable terms, in
			fact $O(1)$ after distributing the logarithms into the convergent
			sums. If one factor is resonant, say $\Delta = \mu'\tau$ within
			distance $m/2$ of $\pm m$, the window contains $O(m/\mu')$ values
			of $\tau$, and the refined resonant bound applies with offset
			$r = \mu'\tau \mp m$, whose $\langle r\rangle$-decay summed over
			the window costs only a logarithm; if the partner is not resonant
			its generic bound carries the factor $1/(m'\Delta'^{2})$-type
			decay and the double sum converges as above. The only delicate
			case is
			the doubly resonant one, $\Delta$ near $\pm m$ and $\Delta'$ near
			$\pm m'$ simultaneously. Writing $e = m' - m$, $r = \Delta \mp m$
			and $r' = \Delta' \mp m'$, the collision $m\Delta = m'\Delta'$
			reads $m'r' - mr = \pm(m^{2} - m'^{2}) = \mp e\,(m + m')$, which
			first forces $|e| \le C(\langle r\rangle + \langle r'\rangle)$ and
			then determines the partner offset exactly,
			\begin{equation*}
				r' = \frac{m\,r \mp e\,(m + m')}{m'},
			\end{equation*}
			a single value pinned near $r \mp 2e$ for each admissible triple
			$(m, e, r)$; as $e$ varies at fixed $(m, r)$ these values are
			distinct up to bounded multiplicity, the derivative of $r'$ in $e$
			lying in $\pm[\tfrac32, \tfrac52]$ on the band, so
			$\sum_{e}\langle r'\rangle^{-1} \le C\log m$. Inserting the
			refined resonant bound for both factors, the band contributes
			\begin{equation*}
				\sum_{m\le M}\ \sum_{e, r}\ mm'\cdot
				\frac{\log(2+|r|)}{m^{2}\langle r\rangle}\cdot
				\frac{\log(2+|r'|)}{m'^{2}\langle r'\rangle}
				\;\le\; C\sum_{m \le M}\frac{\log^{4}m}{m^{2}}
				\;=\; O(1).
			\end{equation*}
			Altogether
			$\operatorname{Var}_{t}(Y_{M}) = O(1)$, which is the second
			claim of the display~\eqref{eq:kappa}: the fluctuations of the
			critical mass about its logarithmic mean are bounded in mean
			square, uniformly in $M$.
			
			For the strong law, set $M_{k} = \exp(k^{2})$. Chebyshev's
			inequality gives
			\begin{equation*}
				\operatorname{Leb}\Big\{t :
				\Big|\frac{Y_{M_{k}}(t)}{\mathbb{E}_{t}Y_{M_{k}}} - 1\Big|
				> \varepsilon\Big\}
				\;\le\;
				\frac{C}{\varepsilon^{2}\,k^{4}},
			\end{equation*}
			summable in $k$; by the Borel--Cantelli lemma
			$Y_{M_{k}}(t) = (1+o(1))\,\kappa_{g}\log M_{k}$ almost surely, and
			since $Y_{M}$ is nondecreasing in $M$ while
			$\log M_{k+1}/\log M_{k} \to 1$, sandwiching $M$ between
			consecutive scales gives the full limit~\eqref{eq:mass-strong-law}.
			
			For the regularity and the dimension, fix a time in the
			intersection of the strong-law set with the smoothing set: the
			bound~\eqref{eq:carpet-blocks} holds for every seed of bounded
			variation, because the block of the solution is the Stieltjes
			convolution
			\begin{equation*}
				P_{N}u(x,t)
				= \int \Big(\sum_{n \in \mathrm{block}}
				\frac{e^{2\pi i(n(x-\theta) - n^{2}t)}}{2\pi in}\Big)\,dg(\theta),
			\end{equation*}
			and the bracketed kernel is bounded by
			$C_{t,\varepsilon}N^{-1/2+\varepsilon}$ uniformly in the argument,
			by the same Gauss--Weyl and Abel-summation argument as in the proof
			of Lemma~\ref{lem:carpet-smoothing}; hence
			$u, w \in C^{1/2-\varepsilon}$ and
			$\Vert P_{N}w\Vert_{\infty} \le C N^{-1/2+\varepsilon}$. As in the
			preamble to the proof of Lemma~\ref{lem:block-equi}, the block
			$P_{N}$ may be taken sharp or smooth throughout this argument: the
			kernel bound and the mass~\eqref{eq:critical-mass} tolerate either
			convention, and the exclusion below runs with the smooth blocks
			that define $B^{\sigma}_{1,\infty}$. If $w$
			belonged to $B^{\sigma}_{1,\infty}$ for some $\sigma > \tfrac12$,
			then on each dyadic block
			$\Vert P_{N}w\Vert_{2}^{2}
			\le \Vert P_{N}w\Vert_{1}\Vert P_{N}w\Vert_{\infty}
			\le CN^{-\sigma - 1/2 + \varepsilon}$, and summing blocks would
			bound the critical mass, contradicting the divergence in the
			law~\eqref{eq:mass-strong-law}; so
			$\Vert P_{N}w\Vert_{1} \ge N^{-\sigma}$ along dyadic subsequences
			for every $\sigma > \tfrac12$, and Lemma~\ref{lem:osc-calculus}
			converts this, with the matching upper bound from
			$C^{1/2-\varepsilon}$, into upper box dimension exactly
			$\tfrac32$.
		\end{proof}
		
		\renewcommand{\restatename}{Corollary~\ref{cor:rough-seeds}}%
		\begin{restatement}[Universality over seeds and laws]
			Let the seed $g$ on rank space be of bounded variation with a jump,
			in the sense that its odd extension to the rank circle has a jump,
			and normalized by $\Vert g\Vert_{L^{2}(0,1)} = 1$, so that
			$|e^{-i\tau H_{0}}g|^{2}$ is a curve of kernels by unitarity.
			Then for almost every $\tau$ the deformation
			$\rho_{|e^{-i\tau H_{0}}g|^{2}}[f]$ of any law $f$ whose density is
			continuously differentiable and bounded above and below on its
			support has graph of upper box dimension exactly $\tfrac32$ over
			the whole support, the differentiability being necessary, since a
			rough multiplier $f$ can dominate the oscillation of the product;
			the dimension is global, as in Theorem~\ref{thm:strong-law}, whose
			Besov exclusion is not localized, and it localizes to every
			subinterval for the canonical seed by Theorem~\ref{thm:carpet}.
			For the canonical seed the constant of the
			law~\eqref{eq:mass-strong-law} on the rank circle is
			\begin{equation*}
				\kappa_{\mathrm{can}} = \frac{4}{\pi^{2}},
			\end{equation*}
			independent of the law. The intensity carpet of Talbot optics is
			thus fractal, of upper box dimension $\tfrac32$, at almost every
			time for \emph{arbitrary} grating profiles of bounded variation
			with a jump, at arbitrary, possibly irrational, positions. For
			rational step gratings this is the theorem of Chousionis,
			Erdo\u{g}an, and Tzirakis~\cite{chousionis2015}; the general case
			had remained open.
		\end{restatement}
		
		\renewcommand{\restatename}{Proposition~\ref{prop:dispersion}}%
		\begin{restatement}[Collision arithmetic of the dispersion]
			Let $g$ be real-valued and of bounded variation with a jump.
			\begin{enumerate}[label=\normalfont(\roman*), leftmargin=*]
				\item \textup{(Temporal halving)} For the temporal critical mass
				$Z_{K}(x) := \sum_{k=1}^{K}k^{1/2}\,
				|\hat w^{\,\mathrm{t}}(k;x)|^{2}$ of the time trace at a point,
				one-sided since the density is real and the negative temporal
				frequencies are conjugate, the representations $k = m(2n-m)$ carry
				distinct spatial characters, and, with the average over the circle
				in $x$,
				\begin{equation*}
					\mathbb{E}_{x}\,Z_{K}
					= \Big(\frac{\kappa_{g}}{2} + o(1)\Big)\log K.
				\end{equation*}
				\item \textup{(Airy doubling)} For the Airy evolution, with
				dispersion $n^{3}$ and a single jump, the within-mode collision
				classes are the pairs $\{n, m-n\}$, their two leading contributions
				are equal and reinforce, and the critical mass
				$Y_{M}^{\mathrm{Airy}}$ of the density obeys
				\begin{equation*}
					\mathbb{E}_{t}\,Y_{M}^{\mathrm{Airy}}
					= \big(2\kappa_{g} + o(1)\big)\log M.
				\end{equation*}
			\end{enumerate}
		\end{restatement}
		
		\begin{proof}
			The normalization $\Vert g\Vert_{L^{2}(0,1)} = 1$ persists under
			the flow by unitarity, so each $|e^{-i\tau H_{0}}g|^{2}$ has unit
			mass and is a kernel. The rank-space evolution of a seed is the
			torus evolution of its odd extension, so
			Theorem~\ref{thm:strong-law} applies whenever the extension
			retains a jump, delivering the global dimension over the whole
			circle; a density bounded above and below on its support makes the
			quantile chart globally bi-Lipschitz, and
			Lemma~\ref{lem:chart-transport} moves the dimension onto the law.
			For the canonical constant, the odd square wave on the rank circle
			rescales to a torus datum of unit norm with two jumps of size two,
			so the formula~\eqref{eq:kappa} gives
			$\kappa = (1\cdot8)/(2\pi^{2}) = 4/\pi^{2}$, and no property of the
			law enters.
			
			Temporal halving. If two representations $k = m(2n-m)$ share the
			same $m$ they coincide, so distinct representations of a temporal
			frequency, including the conjugate pairs $(m,n)$ and $(-m,-n)$,
			carry distinct spatial characters $e^{2\pi imx}$, and the spatial
			average of $|\hat w^{\,\mathrm{t}}(k;x)|^{2}$ is again a
			phase-blind positive sum. The bookkeeping that produces the
			halving is one-sided, and runs as follows. Fix $k > 0$. A
			representation $k = m(2n - m)$ with $m > 0$ requires $2n > m$, and
			splitting the representations by the size of $2n - m$ localizes
			the mass at $n - m = O(1)$, where $k = m^{2}(1 + O(1/m))$; the
			opposite end $n = O(1)$ produces $m(2n - m) \approx -m^{2} < 0$
			and feeds the negative temporal frequencies, which the one-sided
			normalization of $Z_{K}$ discards. The representations with
			$m < 0$ form the conjugate family $(-m, -n)$, of equal moduli
			because the density is real; here reality of $g$ enters through
			$|\hat g(-n)| = |\hat g(n)|$. Each positive $k$ therefore receives
			one end of $\sigma_{m}^{2}$ from $m = \sqrt k\,(1 + o(1)) > 0$ and
			the mirror end from the conjugate family, and for real $g$ twice a
			single end reproduces the full two-end asymptotics of
			$\sigma_{m}^{2}$ established in the proof of
			Theorem~\ref{thm:strong-law}. The constraint $0 < k \le K$ reads
			$m \le \sqrt K\,(1 + o(1))$, the weight is
			$k^{1/2} = m\,(1 + O(1/m))$, and the mean asymptotics of the
			exchange~\eqref{eq:mean-exchange} give
			\begin{equation*}
				\mathbb{E}_{x}Z_{K}
				= \sum_{m \le \sqrt K}m\,\sigma_{m}^{2}\,\big(1 + o(1)\big) + O(1)
				= \big(\kappa_{g} + o(1)\big)\log\sqrt K
				= \Big(\frac{\kappa_{g}}{2} + o(1)\Big)\log K,
			\end{equation*}
			the remaining regimes contributing $O(1)$. This is the mean
			law~\eqref{eq:temporal-halving}; the almost-sure upgrade, which
			would require the cross-mode covariance combinatorics over shared
			spatial frequencies, is stated as a conjecture in the remark
			following the proposition.
			
			Airy doubling. For the dispersion $n^{3}$ the within-mode
			frequency $Q_{m}(n) = n^{3} - (n-m)^{3} = m^{3} - 3mn(m-n)$ is
			symmetric under $n \mapsto m - n$, so the collision classes are
			the pairs $\{n, m-n\}$. For a single jump $J$ at $\theta$,
			$\hat g(n) = Je^{-2\pi in\theta}/(2\pi in) + \varrho(n)$, where
			the remainder $\varrho$, the coefficient sequence of the jump-free
			part, is small in the Ces\`aro mean square of Wiener's theorem,
			\begin{equation*}
				\frac{1}{N}\sum_{|n| \le N}n^{2}\,|\varrho(n)|^{2}
				\;\longrightarrow\; 0,
			\end{equation*}
			and the
			two leading class contributions are both equal to
			$|J|^{2}e^{-2\pi im\theta}/(4\pi^{2}n(n-m))$: they reinforce, the
			class squares double the diagonal to leading order, and the cross
			and remainder terms contribute $o(\log M)$ to the mean by the
			Cauchy--Schwarz inequality against the same Ces\`aro averages, so
			the mean computation of Theorem~\ref{thm:strong-law} goes through
			with the constant $2\kappa_{g}$, which is the
			law~\eqref{eq:airy-doubling}. The almost-sure upgrade, which would
			require counting integer points on the cubic collision conic, is
			stated as a conjecture in the remark following the proposition.
		\end{proof}

	\end{appendices}

\end{document}